\documentclass[journal,10pt]{IEEEtran}               
\usepackage[T1]{fontenc}
\usepackage{amsmath}
\usepackage{amssymb}
\usepackage{stmaryrd}
\usepackage{epic,eepic}
\usepackage{theorem}
\usepackage{pifont}  
\usepackage{euscript}
\usepackage{calc}
\usepackage{cite}
\usepackage[titles]{tocloft}

\usepackage[usenames]{color}          
\usepackage{xcolor}
\definecolor{Brown}{rgb}{0.55,0.0,0.10}
\definecolor{dgreen}{rgb}{0.00,0.56,0.00}
\definecolor{vertmoinsfonce}{rgb}{0.00,0.35,0.00}
\definecolor{vertclair}{rgb}{0.00,0.85,0.00}
\definecolor{llightggray}{rgb}{0.97,0.97,0.97}
\definecolor{lightggray}{rgb}{0.9,0.9,0.9}
\definecolor{ggray}{rgb}{0.5,0.5,0.5}
\definecolor{darkggray}{rgb}{0.25,0.25,0.25}
\definecolor{ddarkggray}{rgb}{0.1,0.1,0.1}
\definecolor{bleu}{rgb}{0.00,0.00,0.00}

\usepackage{shortvrb,psfrag}
\usepackage{epsf}           
\usepackage{graphicx}       
\usepackage{epsfig}         
\usepackage{pstricks}
\usepackage{pstricks-add}
\usepackage{pst-plot}
\usepackage{dsfont}
\theoremheaderfont{\color{black}\normalfont\bfseries}

\newtheorem{lemma}{Lemma}
\newtheorem{theorem}{Theorem}
\newtheorem{definition}[theorem]{Definition}
\newtheorem{proposition}[theorem]{Proposition}

\theoremstyle{plain}{\theorembodyfont{\rmfamily}%
}
\theoremstyle{plain}{\theorembodyfont{\rmfamily}%
\theoremstyle{plain}{
\theorembodyfont{\rmfamily}

	\newtheorem{remark}[theorem]{Remark}
	
	}

\newcommand{\ie}{i.e., }

\newcommand{\R}{\mathbb{R}}

\newcommand{\N}{\mathbb{N}}

\newcommand{\E}{\mathbb{E}}

\newcommand{\C}{\mathcal{C}}

\newcommand{\QQ}{\mathcal{Q}}

\newcommand{\D}{\mathcal{D}}

\newcommand{\PP}{\mathcal{P}}

\newcommand{\mc}{\mathcal}

{\Large\normalsize}
{\Large\normalsize}

\ifCLASSINFOpdf
\else
\fi
\hyphenation{op-tical net-works semi-conduc-tor}

\begin{document}
%
\title{Rate Adaptation for Secure HARQ Protocols}



\author{\IEEEauthorblockN{Ma\"{e}l Le Treust,
Leszek Szczecinski\IEEEauthorrefmark{1} and Fabrice Labeau\IEEEauthorrefmark{2}}\\
\IEEEauthorblockA{
ETIS UMR 8051, Universit\'{e} Paris Seine, Universit\'{e} Cergy-Pontoise, ENSEA, CNRS,\\
6, avenue du Ponceau, 95014 Cergy-Pontoise CEDEX, FRANCE\\
\IEEEauthorblockA{\IEEEauthorrefmark{1}
INRS, Montreal, Canada}\\
\IEEEauthorblockA{\IEEEauthorrefmark{2}
McGill University, Montreal, Canada}\\
mael.le-treust@ensea.fr, leszek@emt.inrs.ca, fabrice.labeau@mcgill.ca}

\thanks{Work supported by the government of Quebec under grant \#PSR-SIIRI-435 and SRV ENSEA 2014; conducted as part of the project Labex MME-DII (ANR11-LBX-0023-01); presented in part at the IEEE Information Theory Workshop, Sept. 2013  \cite{LeTreustSzczecinskiLabeau13}. This work was carried out, in part, when Ma\"{e}l Le Treust was a post-doctoral researcher with INRS and McGill University.}}



\maketitle

\begin{abstract}
This paper investigates the incremental-redundancy hybrid-automatic repeat request (IR-HARQ) transmission over  independent block-fading channels in the presence of an eavesdropper, where the secrecy of the transmission is ensured via introduction of dummy-messages. Since the encoder only knows the statistics of the channel state, the secrecy and the reliability are defined in a probabilistic framework. Unlike previous works on this subject, we design a coding strategy tailored to IR-HARQ by splitting the dummy-message rate over several rate parameters. These additional degrees of freedom improve the match between the dummy-message rates and the realizations of the eavesdropper channels. We evaluate the performance in terms of secrecy outage probability,  connection outage probability and throughput and we compare it with the benchmark paper by Tang et al. \cite{TangLiuSpasojevicPoor09}. Numerical examples illustrate that, comparing to existing alternatives, splitting of the dummy-message rate provides higher throughput and lower expected duration/average delay.
\end{abstract}




%
\IEEEpeerreviewmaketitle

\section{Introduction}

This work is concerned with the transmission of information over wireless independent block-fading channels, where the channel state information (CSI), which captures the essence of channel statistics, is not available at the transmitter but can be estimated by the receivers. In such a scenario, the transmission is inherently i)~unreliable due to unpredictable fading, and ii)~unsecure due to the possibility of eavesdropping when communicating over a broadcast medium. The successful communication and the secrecy can thus only be defined/guaranteed in probabilistic terms. The principal question we want to investigate  is how the constraints on the secrecy and the reliability are related when transmissions are carried out using an incremental-redundancy hybrid-automatic repeat request (IR-HARQ) protocol, and how to construct the coding to take advantage of the additional dimension offered by retransmissions.

\subsection{State of art}\label{Sec:State.of.Art}
\textbf{Reliability and IR-HARQ}\\
Reliability is a key issue in modern communications and is deeply related to the knowledge---by the transmitters---of the channel statistics often summarized in one parameter, which defines the  CSI, e.g.,  the signal-to-noise ratio (SNR). When both encoder and decoder know the CSI it is possible to design an appropriate coding scheme that conveys information with arbitrary reliability \cite{shannon-bell-1948}. When the CSI is unavailable at the transmitter, the successful transmission cannot be guaranteed leading to the concepts of  outage probability and throughput.

To deal with unavoidable transmission errors, the so-called hybrid automatic repeat request (HARQ) protocol is often used:  a single-bit acknowledgement feedback (\textsf{Ack}/\textsf{Nack}) indicates whether the decoding was successful or not. Then, the transmitter may transmit the same message many times, till it is successfully received--the event indicated by the \textsf{Ack}. The two main classes of HARQ protocols are i) Repetition Time Diversity (RTD), which consists in repeated transmission of the same codeword, and ii) Incremental Redundancy (IR), a more powerful scheme which involves a different codebook in each transmission. HARQ protocols were analyzed in the literature from the point of view of throughput, outage probability, and average delay, e.g., \cite{TelatarGallager95,CaireTuninettiHARQ2001,ZorziRaob96,ZorziRao97,ZorziBorgonovo97, Zorzi98}.

Retransmissions in HARQ provide additional degrees of freedom which can be exploited to design a code which provides a suitable ``match'' between the transmission rate and the channel realizations. For example, in \cite{VisotskyYakunTripathiHonigPeterson05,PfletschingerNavarro10,UhlemannRasmussenGrantWiberg10, SzczecinskiKhosraviradDuhamelRahman13,JabiBenjillaliSzczecinskiLabeau16, JabiHamssSzczecinskiPiantanida15, JabiSzczecinskiBenjillaliLabeau15, Larsson16, saee,Lee15}, the length of codewords was varied throughout the retransmissions. A different approach was taken by \cite{Hausl07, Chui07, Duyck10, Trillingsgaard14, NguyenTimo15, BenyoussJabiLeTreustSzczecinski16, JabiBenyoussLeTreustDoraySzczecinski16} which kept the codeword length constant and rather relied on the design of new coding schemes to increase the throughput.

\textbf{Secrecy}\\
Security is an issue in wireless communications due to the broadcast nature of the transmission medium. An eavesdropper within the communication range can ``overhear''  the transmitted signals and extract some private information. 

Instead of using cryptographic methods to protect the message, Wyner  \cite{Wyner(Wiretap)1975} proposed to exploit the difference between the legitimate decoder and the eavesdropper channels, and characterized the rate at which the legitimate users can communicate not only reliably but also securely. The \emph{threat model} of \cite{Wyner(Wiretap)1975} refers to the one of Shannon's cipher system \cite{Shannon(secrecy)1949}, defined in an information-theoretic sense, also referred to as ``Physical Layer Security''. In particular, the eavesdropper is assumes to have arbitrary equipment and computing power and to know the existence of a message intended to the legitimate receiver \cite[pp. 656]{Shannon(secrecy)1949}. The eavesdropper is aware of the code-book used in encoding and decoding operations \cite[pp. 1355]{Wyner(Wiretap)1975}. Perfect secrecy is defined by the condition that, for any eavesdropper's observation the a-posteriori probabilities are equal to the a-priori probabilities \cite[pp. 679]{Shannon(secrecy)1949}. The goal is to exploit the intrinsic randomness of the channel in order to secure the transmission.\footnote{Note that we use here secrecy, as the only mean to guarantee the security of a transmission, but secrecy can be combined with cryptography as well.} 

The results of \cite{Wyner(Wiretap)1975} were further generalized in \cite{CsiszarKorner(BroadcastConf)78, LeungHellman(GaussianWiretap)78} under assumption of CSI knowledge, which has a significant impact on security in wireless networks  \cite{BlochBarros11}. In \cite{BlochBarrosRodriguesMcLaughlin08}, the authors proved that secure communication is possible even when the eavesdropper has, on average, a channel stronger than that of the receiver. However, the legitimate users must have perfect knowledge of their CSI and estimate the CSI of the eavesdropper. 
In \cite{KhistiTchamkertenWornell08}, the problem of broadcasting confidential messages to multiple receivers over parallel and fast-fading channels was investigated while \cite{GopalaLaiElGamal08} characterizes the secrecy capacity of slow-fading wiretap channel under different CSI assumptions. The ergodic secrecy capacity was characterized in \cite{LiangPoorShamai08} assuming full CSI at both legitimate transmitters.

The assumption of the knowledge of the eavesdropper's CSI is an idealization,\footnote{There is no reason that eavesdropper would collaborate with the legitimate users.} so \cite{LiYatesTrappeTWC10} studied the case where the channel to the eavesdropper experiences fading not known to the legitimate users. The effect of partial CSI on  achievable secure communication rates and on secret-key generation was also investigated in \cite{BlochLaneman13}, and \cite{RezkiKhistiAlouini14} provided bounds on the ergodic secrecy capacity. The case of transmission without CSI at the encoder was investigated in \cite{LinJorswieck16}, where the ergodic secrecy capacity for fast fading wiretap channel was characterized; and in  \cite{ZhouMcKayMahamHjorungnes11}, which proposed an alternative secrecy outage formulation to measure the probability that message transmission fails to achieve perfect secrecy.

\textbf{Secrecy and HARQ}\\
Retransmissions in HARQ may be used not only to increase the reliability or the throughput, but also to increase the secrecy. This issue was  investigated  in \cite{TangLiuSpasojevicPoor09} using extension of the Wyner code \cite{Wyner(Wiretap)1975} with the introduction of dummy messages. In the absence of  CSI, the coding parameters were chosen using the statistics of the CSI. Then, receiving a \textsf{Nack} feedback, the encoder retransmits the message but has no guarantee of reliability nor secrecy which are then characterized via the random events of secrecy outage and connection outage. Improvement of the secure HARQ protocol was investigated in \cite{MheichLeTreustAlbergeDuhamelSzczecinski14}, \cite{MheichLeTreustAlbergeDuhamel16} with variable-length coding and in \cite{BaldiBianchChiaraluce12} using low-density parity-check (LDPC) codes. In \cite{TomasinLaurenti14}, the authors investigate secure HARQ protocols based on multiple encoding, by using new dummy-messages at each transmission. In \cite{Choi17}, the author exploits the channel reciprocity assumption in order to transmit securely even if the channel to the eavesdropper is less noisy than the channel to the legitimate decoder. In that case, the channel state information shared by the pair of legitimate transmitters can be used as a secret key.



It is worthwhile to mention that the notion of secrecy may be defined in many different ways, including ``perfect'', ``weak'' and ``strong secrecy'' \cite{BlochBarros11}, ``effective secrecy'', ``privacy'' and ``stealth'' \cite{HouKramer14}, ``semantic security'' \cite{GoldfeldCuffPermuter16},  \cite{SenigagliesiBaldiChiaraluce17}, or ``covert communications'' \cite{WangWornellZheng16}, \cite{Bloch16}. Each of these notions provides different degrees of secrecy, based on probabilistic arguments or worst case scenarios. 

The goal of this work is not to investigate the comparison between these different notions but rather to develop a coding scheme tailored for HARQ transmissions. We use the same secrecy metric, namely ``weak secrecy'', as in the previous articles on that subject \cite{TangLiuSpasojevicPoor09}, \cite{TomasinLaurenti14} which also follow Wyner's work \cite{Wyner(Wiretap)1975}.

In this paper we investigate the canonical model of independent block-fading channels and we focus only on IR-HARQ protocol because it offers new degrees of freedom in the code design; on the other hand, these degrees of freedom are, by definition, absent from the RTD coding.

\subsection{Contributions and organizations}
A natural trade-off arises between reliability and security in the wiretap channel: when the dummy-message rate  increases, it decreases the secrecy outage probability but  increases the connection outage probability. One important drawback of the coding schemes proposed in \cite{TangLiuSpasojevicPoor09}, is that the dummy-message rate is unique and should guarantee the secrecy for a large number of possible transmissions, even if the expected duration/average delay of the transmission is much lower. In this work, we address this issue upfront and design an original wiretap code by splitting the dummy-message rate over several rate parameters. These additional degrees of freedom improve the match between the dummy-message rates and the  realization of the eavesdropper channels. The article \cite{TangLiuSpasojevicPoor09} is clearly the benchmark for this work. Our contributions are the following:
\begin{itemize}
\item We propose a novel wiretap code, called ``Adaptation-Secrecy-Rate-code'' (ASR-code) that splits the dummy-message into multiple dummy-messages and inserts them into upcoming packets. We prove that ASR-code has an arbitrarily small error probability and an arbitrarily small information leakage rate, for a whole set of channel states. In our view, the ASR-code generalizes the coding scheme presented in \cite{TangLiuSpasojevicPoor09} in a very natural manner as, for a particular choice of the dummy-message rates, the ASR-code is equivalent to the coding proposed in \cite{TangLiuSpasojevicPoor09}.
\item We characterize the trade-off between connection and secrecy outage probabilities and show the optimal rate allocation for discrete channels and for Rayleigh fading channels with one transmission. 
\item We present a numerical optimization for multiple transmissions over Rayleigh fading channel: using the splitting of the dummy-message rate, we achieve a higher throughput with a lower expected duration/average delay.
\item ASR-code provides better performances than the protocols of \cite{TangLiuSpasojevicPoor09} and \cite{TomasinLaurenti14} for discrete and Gaussian channels.
\end{itemize}

The main differences with our previous work  \cite{LeTreustSzczecinskiLabeau13} are: 
\begin{itemize}
\item We consider an arbitrary number of possible retransmissions, whereas only one retransmission was considered in \cite{LeTreustSzczecinskiLabeau13}; this affects non-trivially the expressions of connection and secrecy outages probabilities.
\item We consider a more practical case of Rayleigh block-fading channels and analyze the corresponding solutions. 
\item We provide a full version of the proof of Theorem \ref{theo:RandomLeszekCode}, while only a sketch was shown in 
\cite{LeTreustSzczecinskiLabeau13}.
\end{itemize}

 The work is organized as follows. Sec.~\ref{sec:SecureHARQ} presents the channel model under investigation, the HARQ-code and defines our new protocol called ASR-code. The main result is Theorem \ref{theo:RandomLeszekCode} which proves that the error probability and the information leakage rate converge to zero for large block length. The performance of the ASR-code is measured by the secrecy throughput and the secrecy/connection outage probability,  defined in Sec.~\ref{sec:DefThroughputOutages}. The example of a discrete channel state is shown in Sec.~\ref{sec:DiscreteChannels} whereas Rayleigh fading channels are investigated in Sec.~\ref{sec:RayleighGaussian}. Sec.~\ref{sec:conclusion} concludes the paper and the proofs of the results are stated in the Appendix.




\section{Secure HARQ Protocol}\label{sec:SecureHARQ}

We consider a HARQ protocol with $L$ possible transmissions shown schematically in Fig.~\ref{fig:StateDepParallelWiretapsState} for $L=2$. Each transmission $l \in \{1,\ldots,L\}$ corresponds to a block of $n\in\N$ symbols. Capital letter $X$ denotes the random variable, lowercase letter $x\in\mc{X}$ denotes the realization and $\mc{X}^n$ denotes the $n$-time Cartesian product of the set $\mc{X}$. The random message $M\in\mc{M}$ is uniformly distributed and $m\in\mc{M}$ denotes the realization.

 During the first transmission, the encoder $\C$ uses the sequence of input symbols $x_1^n \in \mc{X}^n$ in order to transmit the message $m\in \mc{M}$ to the legitimate decoder $\D$. The decoder $\mc{D}$ (resp. eavesdropper $\mc{E}$)   observes the sequence of channel outputs $y_1^n \in \mc{Y}^n$ (resp. $z_1^n \in \mc{Z}^n$) and tries to decode (resp. to infer) the transmitted message $m\in \mc{M}$. The decoder $\D$ sends a $\textsf{Ack}_1/\textsf{Nack}_1 $ feedback over a perfect channel that indicates to the encoder, whether the first transmission was correctly decoded or not.  
 
 If the encoder receives a $\textsf{Nack}_{l-1}$ feedback after $l-1 \in \{1,\ldots,L\}$ transmissions, then the message $m\in \mc{M}$ was not correctly decoded yet. The encoder starts retransmitting  the message  $m\in \mc{M}$ over transmission $l \in \{2,\ldots,L\}$  with input sequence $x_{l}^{ n} \in \mc{X}^{ n}$. The decoder $\D$ (resp. eavesdropper $\mc{E}$) tries to decode (resp. to infer) the transmitted message $m\in \mc{M}$ from sequences of channel outputs $(y_1^n, y_2^{ n} ,\ldots y_l^n) \in  \mc{Y}^{l \times n}$  (resp. $(z_1^n  , z_2^n,\ldots z_l^n) \in \mc{Z}^{l \times n}$), where $\mc{Y}^{l \times n}  = \overbrace{\mc{Y}^n \times  \ldots  \times \mc{Y}^n}^\text{$l$}$ is the $l$-time Cartesian self-product of set $\mc{Y}^n$. If the maximal number of transmissions $L$ is attained, the encoder drops message $m \in \mc{M}$ and starts sending the next message $m' \in \mc{M}$. The notation $\Delta(\mc{X})$ stands for the set of the probability distributions $\PP(X)$ over the set $\mc{X}$. We assume that the channel is memoryless  with transition probability $\mc{T}(y,z|x,k)$ depending on a state parameter $k\in \mc{K}$, for example a fading coefficient. The state parameters $(k_1, k_2,\ldots, k_L ) \in \mc{K}^L $ stay constant during the transmission of a block of $n\in\N$ symbols and are chosen at random with i.i.d. probability distribution $\mc{P}_k \in \Delta(\mc{K})$, from one block to another. The state parameters $(k_1, k_2,\ldots, k_L ) \in \mc{K}^L $ are observed by the decoder and the eavesdropper but not by the encoder. 

At transmission $l\in\{1,\ldots,L\}$, the state-dependent wiretap channel is given by 
\begin{eqnarray}
\mc{T}^{n}(y_l^n,z_l^n|x_l^n,\textcolor[rgb]{0.00,0.00,0.00}{k_l}) = \prod_{i=1}^n\mc{T}(y_l(i),z_l(i)|x_l(i),\textcolor[rgb]{0.00,0.00,0.00}{k_l}),\label{eq:transition}
\end{eqnarray}
where $x_l(i)$ (resp. $y_l(i)$, $z_l(i)$) denotes the $i$-th symbol of the transmission block $x_l$ (resp. $y_l$, $z_l$) of length $n$. The channel statistics are known by both encoder $\C$ and decoder $\D$.

\begin{figure}[tb]
\begin{center}
\psset{xunit=0.9cm,yunit=0.5cm}
\begin{pspicture}(-0.5,-2.8)(8.5,3.4)
\psframe(0,0)(1,2)
\psframe(7,0)(8,2)
\pscircle[linecolor = blue](4.5,1.5){0.45}
\pscircle[linecolor = red](3.5,0.5){0.45}
\psline[linewidth=1pt]{->}(-1,1)(0,1)
\psline[linewidth=1pt]{->}(8,1)(9,1)
\psline[linewidth=1pt,linecolor = red]{->}(1,0.5)(3,0.5)
\psline[linewidth=1pt,linecolor = red]{->}(4,0.5)(7,0.5)
\psline[linewidth=1pt,linecolor = blue]{->}(1,1.5)(4,1.5)
\psline[linewidth=1pt,linecolor = blue]{->}(5,1.5)(7,1.5)
\rput[u](4.5,1.5){$\textcolor[rgb]{0.00,0.00,1.00}{\mc{T}_1}$}
\rput[u](3.5,0.5){$\textcolor[rgb]{1.00,0.00,0.00}{\mc{T}_2}$}
\rput[u](0.5,1){$\mc{C}$}
\rput[u](7.5,1){$\mc{D}$}
\rput[u](-0.5,1.4){$M$}
\rput[u](8.5,1.6){$\hat{M}$}
\rput[u](1.5,1.9){$\textcolor[rgb]{0.00,0.00,1.00}{X_1^n}$}
\rput[u](1.5,0.9){$\textcolor[rgb]{1.00,0.00,0.00}{X_2^{ n}}$}
\rput[u](6.5,1.9){$\textcolor[rgb]{0.00,0.00,1.00}{Y_1^n}$}
\rput[u](6.5,0.9){$\textcolor[rgb]{1.00,0.00,0.00}{Y_2^{ n}}$}
\psframe(7,-3)(8,-1)
\psline[linewidth=1pt,linecolor = red]{->}(3.5,-0.35)(3.5,-2.5)(7,-2.5)
\psline[linewidth=1pt,linecolor = blue]{->}(4.5,0.65)(4.5,-1.5)(7,-1.5)
\rput[u](6.5,-1.1){$\textcolor[rgb]{0.00,0.00,1.00}{Z_1^n}$}
\rput[u](6.5,-2.1){$\textcolor[rgb]{1.00,0.00,0.00}{Z_2^{ n}}$}
\rput[u](7.5,-2){$\mc{E}$}
\psline[linewidth=1pt,linecolor = vertmoinsfonce]{-}(0.5,3.5)(7.7,3.5)
\psline[linewidth=1pt,linecolor = vertmoinsfonce]{-}(7.7,3.5)(7.7,2)
\psline[linewidth=1pt,linecolor = vertmoinsfonce]{->}(0.5,3.5)(0.5,2)
\rput[u](1.6,3.1){\textcolor[rgb]{0.00,0.50,0.00}{\textsf{Ack}/\textsf{Nack}}}
\psline[linewidth=1pt,linecolor = blue]{->}(4.5,2.8)(4.5,2.35)
\rput[u](4.5,3.1){$\textcolor[rgb]{0.00,0.00,1.00}{k_1}$}
\psline[linewidth=1pt,linecolor = red]{->}(3.5,2.8)(3.5,1.35)
\rput[u](3.5,3.1){$\textcolor[rgb]{1.00,0.00,0.00}{k_2}$}
\psline[linewidth=1pt,linecolor = blue]{<->}(7.3,0)(7.3,-1)
\psline[linewidth=1pt,linecolor = red]{<->}(7.5,0)(7.5,-1)
\psdots[linecolor = red](7.5,-0.45)
\psdots[linecolor = blue](7.3,-0.55)
\psline[linewidth=1pt,linecolor = red]{-}(8.5,-0.45)(7.5,-0.45)
\psline[linewidth=1pt,linecolor = blue]{-}(8.5,-0.55)(7.3,-0.55)\
\rput[d](8.8,-0.8){$\textcolor[rgb]{0.00,0.00,1.00}{k_1}$}
\rput[u](8.8,-0.2){$\textcolor[rgb]{1.00,0.00,0.00}{k_2}$}
\end{pspicture}
\caption{State dependent wiretap channels $\mc{T}_i(y_i,z_i|x_i,k_i)$, with $i\in \{1,2\}$.
The second transmission starts if the encoder $\C$ receives a \textcolor[rgb]{0.00,0.00,0.00}{\textsf{Nack}} feedback from the legitimate decoder. The state parameters $\textcolor[rgb]{0.00,0.00,0.00}{k_1 \in \mc{K}_1}$ and $\textcolor[rgb]{0.00,0.00,0.00}{k_2 \in \mc{K}_2}$ are chosen arbitrarily, stay constant during the transmission and are available only at the legitimate decoder $\D$ and at the eavesdropper $\mc{E}$.}
\label{fig:StateDepParallelWiretapsState}
\end{center}
\end{figure}
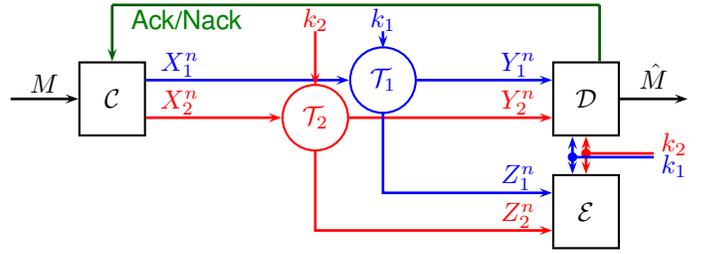

\begin{definition}\label{def:CodeLz}
A HARQ-code $c_n \in\mc{C}(n,\textsf{R},L)$ with stochastic encoder is a vector of encoding and decoding  functions $c_n=\Big( (f_l)_{l\in \{1,\ldots,L\}}, (g_l)_{l\in \{1,\ldots,L\}})$, defined for each transmission $l\in \{1,\ldots,L\}$  as follows:
\begin{eqnarray}
&f_l& : \mc{M}\times  \mc{X}^{(l-1)\times  n} \times \{\textsf{Ack},\textsf{Nack}\}^{l-1} \rightarrow \Delta(\mc{X}^{ n}),\label{eq:Code3Lz}\\
&g_l& : \mc{Y}^{l\times n} \times  \mc{K}^l \longrightarrow {\mc{M}}   \times \{\textsf{Ack},\textsf{Nack}\} ,\label{eq:Code4Lz}
\end{eqnarray}
where the rate $\textsf{R} $ defines the cardinality $ |\mc{M}|= 2^{n \textsf{R}}$ of the set of messages $\mc{M}$ and $L$ is the maximal number of transmissions. We denote by $\mc{C}(n,\textsf{R},L)$, the set of HARQ-codes with stochastic encoder. 
\end{definition}

\begin{definition}\label{def:RandomPropertyCode}
For each vector of state parameters $(k_1, \ldots, k_L ) \in \mc{K}^L$, the error probability $\mc{P}_{\textsf{e}}$ and the information leakage rate $\mc{L}_{\textsf{e}}$ of the  HARQ-code $c_n \in\mc{C}(n,\textsf{R},L)$ are defined by:
\begin{eqnarray*}
\mc{P}_{\textsf{e}}\big({c}_n \big| k_1,\ldots,k_L \big)= \PP\Big(M\neq\hat{M}\Big|\;c_n, k_1,\ldots,k_L \;\Big),\label{eq:DefErrorProbaLz}\\
\mc{L}_{\textsf{e}}\big({c}_n \big| k_1,\ldots,k_L \big)= \frac{I\Big( M ; Z_1^n,\ldots , Z_L^n \Big| \;c_n, k_1,\ldots,k_L\; \Big)}{n}.\label{eq:DefLeakageRateLz}
\end{eqnarray*}
The random variable $\hat{M}$ denotes the output message of the legitimate decoder. Depending on the number of transmissions $l \in \{1,\ldots,L\}$, it is given by $\hat{M} = g_l(Y_1^n,\ldots,Y_l^n ,k_1 ,\ldots, k_l)$. A non-zero leakage rate means that the eavesdropper can infer some information about the message $M$, which is undesirable.
\end{definition}


In \cite{TangLiuSpasojevicPoor09}, the authors prove the existence of a HARQ-code that has small error probability and small information leakage rate for a whole range of channel states $(k_1, \ldots, k_L ) \in \mc{K}^L$. The coding scheme is based on Wyner's coding for the wiretap channel \cite{Wyner(Wiretap)1975} and involves two parameters:  the rate  $ \sf{R_s}\geq0$, which is called the ``secrecy rate'' and corresponds to the amount of secret information to be transmitted to the legitimate decoder; and the rate $ \sf{R_0}\geq0$ which corresponds to the total size of the codebook. The difference $\sf{R_0} -  \sf{R_s}\geq0$ is called the ``dummy-message rate'' and corresponds to the amount of randomness that will be introduced in the codebook, in order to confuse the eavesdropper. 

Then, the conditions which are sufficient for the transmission to be reliable and secure, given by
\begin{eqnarray}
\sf{R_0}  &\leq&\sum_{j\in 1}^L  I(X_j;Y_j|k_j) ,  \label{eq:TransmissionRate}  \\
\sf{R_0}   - \sf{R_s} &\geq&  \sum_{j\in 1}^L  I(X_j;Z_j|k_j), \label{eq:SecrecyRate}
 \end{eqnarray}
define ``the secure channel set'' \cite[Definition 2]{TangLiuSpasojevicPoor09}.

 
We note that \eqref{eq:SecrecyRate} enforces a high value of the dummy-message rate $\sf{R_0}   - \sf{R_s}$ which must guarantee the secrecy for the maximal number of transmissions $L$. This, in turn, prevents  the first transmissions from being reliable, especially when the number of possible transmissions $L$ is large. 

From this observation stems the main contribution of our work which consists in splitting the dummy-message rate $\sf{R_0} - \sf{R_s}$ over $L$ different parameters denoted by $ \textsf{R}_{1}, \textsf{R}_{2}, \ldots , \textsf{R}_{L} $. Splitting the dummy-message rate makes the first transmissions more reliable, since the first dummy-message rates can be smaller than $\sf{R_0} - \sf{R_s}$ in \cite{TangLiuSpasojevicPoor09}.
 The price is paid by a more complex encoding/decoding; also the outage analysis is more involved, since the $L$ dummy-message rate parameters induce $L$ constraints, stated in equations \eqref{eq:channelstates2} - \eqref{eq:channelstates4} of Definition \ref{def:ChannelStatesLz}.

\begin{definition}[Channel States]\label{def:ChannelStatesLz}
For a fixed number of transmissions $l\in \{1 , \ldots, L \}$, fixed parameters $\varepsilon$, $\textsf{R}$, $ \textcolor[rgb]{0.00,0.00,0.00}{\textsf{R}_{1}}, \ldots , \textcolor[rgb]{0.00,0.00,0.00}{\textsf{R}_{L}}$ and a fixed probability distributions $\PP_{\sf{x}}^{\star}  \in \Delta(\mc{X})$, the set of secure channel states, denoted by $\mc{S}_l( \varepsilon,  \textsf{R} , \textcolor[rgb]{0.00,0.00,0.00}{\textsf{R}_{1}}, \ldots , \textcolor[rgb]{0.00,0.00,0.00}{\textsf{R}_{L}} ,  \PP_{\sf{x}}^{\star})$, is the union of channel states $(k_1,\ldots, k_l) \in \mc{K}^l$ that satisfy the following set of  equations:
\begin{eqnarray}
\textsf{R} +  \sum_{j= 1}^l \textcolor[rgb]{0.00,0.00,0.00}{\textsf{R}_{j}} &\leq&   \sum_{j= 1}^l  I(X_j;Y_j|k_j)  - \varepsilon, \label{eq:channelstates1}\\
\sum_{j= 1}^l \textcolor[rgb]{0.00,0.00,0.00}{\textsf{R}_{j}} &\geq&   \sum_{j= 1}^l  I(X_j;Z_j|k_j) - \varepsilon,  \label{eq:channelstates2} \\
\sum_{j= 1}^{l-1} \textcolor[rgb]{0.00,0.00,0.00}{\textsf{R}_{j}} &\geq&   \sum_{j= 1}^{l-1}  I(X_j;Z_j|k_j) - \varepsilon,  \label{eq:channelstates3} \\
&\vdots&\nonumber \\
\textcolor[rgb]{0.00,0.00,0.00}{\textsf{R}_{1}}&\geq&   I(X_1;Z_1|k_1) - \varepsilon. \label{eq:channelstates4} 
 \end{eqnarray}
\end{definition}
Equation \eqref{eq:channelstates1} guarantees the correct decoding whereas equations \eqref{eq:channelstates2} - \eqref{eq:channelstates4} guarantee that the secrecy condition is satisfied at each transmission $l = \{1, \ldots, L\}$. We note that \eqref{eq:channelstates1} - \eqref{eq:channelstates4} generalize equations of \cite{TangLiuSpasojevicPoor09}. That is, using $\textsf{R}_{2}= \ldots =  \textsf{R}_{L} = 0$, $\textsf{R}_{1} = \textsf{R}_{\sf{0}}  -  \textsf{R}_{\sf{s}} $ and $\textsf{R}  = \textsf{R}_{\sf{s}} $ we obtain \eqref{eq:TransmissionRate} and \eqref{eq:SecrecyRate}.\footnote{Where, formally, $\varepsilon$ should also be added as in \eqref{eq:channelstates1} - \eqref{eq:channelstates4}.}

These conditions are represented graphically in Fig.~\ref{fig:RateRegions} for $L=2$ transmissions.

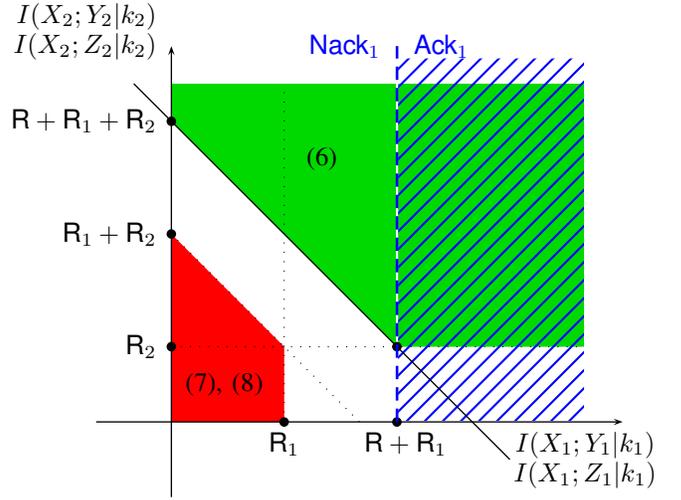
\begin{figure}[!ht]
\begin{center}
\psset{xunit=1cm,yunit=1cm}
\begin{pspicture}(-1,-1)(5,5.5)
\pspolygon*[linecolor=vertclair](0,4.5)(0,4)(3,1)(5.5,1)(5.5,4.5) 
\pspolygon*[linecolor=red](0,0)(0,2.5)(1.5,1)(1.5,0)
\pspolygon[fillstyle=hlines,hatchcolor=blue,linecolor=white](3,0)(5.5,0)(5.5,4.85)(3,4.85)
\psline[linewidth=0.5pt]{->}(-1,0)(6,0)
\psline[linewidth=0.5pt]{->}(0,-1)(0,5)
\rput[d](5.5,-0.3){$I(X_1;Y_1|k_1)$}
\rput[d](5.5,-0.7){$I(X_1;Z_1|k_1)$}
\rput[r](-0.2,5.4){$I(X_2;Y_2|k_2)$}
\rput[r](-0.2,5){$I(X_2;Z_2|k_2)$}
\psline[linewidth=0.5pt](-0.5,4.5)(4.5,-0.5)
\psdots(1.5,0)(0,4)(0,1)(0,2.5)(3,0)(3,1)
\rput[d](1.5,-0.3){$\textsf{R}_{1}$}
\rput[u](2,3.5){\eqref{eq:channelstates1}}
\rput[u](0.7,0.5){\eqref{eq:channelstates2}, \eqref{eq:channelstates3}}
\rput[r](-0.2,1){$\textsf{R}_{2}$}
\rput[r](-0.2,2.5){$\textsf{R}_{1} + \textsf{R}_{2}$}
\rput[r](-0.2,4){$\textsf{R} + \textsf{R}_{1} + \textsf{R}_{2}$}
\psline[linewidth=0.5pt,linestyle=dotted](0,1)(5.5,1)
\psline[linewidth=0.5pt,linestyle=dotted](1.5,0)(1.5,4.5)
\psline[linewidth=0.5pt,linestyle=dotted](2.5,0)(0,2.5)
\psline[linewidth=1pt,linestyle=dashed,linecolor = blue](3,0)(3,5)
\rput[d](3,-0.3){ $\; \textsf{R} + \textsf{R}_{1}$}
\rput[l](3,5){ $\;\textcolor[rgb]{0.00,0.00,1.00}{\textsf{Ack}_1}$}
\rput[r](3,5){$ \textcolor[rgb]{0.00,0.00,1.00}{\textsf{Nack}_1} \;$ }
\end{pspicture}
\caption{Decoding and secrecy regions corresponding to the rates $(\textsf{R} ,  \textsf{R}_{1}, \textsf{R}_{2})$, for $L=2$ transmissions. The second transmission starts only if there is a $\textsf{Nack}_1$, hence we disregard the dashed region of $\textsf{Ack}_1$. The green upper region corresponds to the decoding constraint of equation \eqref{eq:channelstates1} for the mutual informations $I(X_1;Y_1|k_1)$ and $I(X_2;Y_2|k_2)$. The red lower region corresponds to the secrecy constraints of equations \eqref{eq:channelstates2}, \eqref{eq:channelstates3} for the mutual informations $I(X_1;Z_1|k_1)$ and $I(X_2;Z_2|k_2)$.}
\label{fig:RateRegions}
\end{center}
\end{figure}

We now prove the existence of a HARQ-code such that the error probability $\mc{P}_{\textsf{e}}$ and the information leakage rate $\mc{L}_{\textsf{e}}$ converge to zero, for all tuples of channel states $(k_1,\ldots, k_L)$ that belong to  $\bigcup_{l = 1}^L \mc{S}_l( \varepsilon,  \textsf{R} , \textsf{R}_{1}, \ldots , \textsf{R}_{L} ,  \PP_{\sf{x}}^{\star})$.

\begin{theorem}[Compound Wiretap Channel]\label{theo:RandomLeszekCode}
Fix the parameters $\textsf{R}$, $ \textcolor[rgb]{0.00,0.00,0.00}{\textsf{R}_{1}}, \ldots , \textcolor[rgb]{0.00,0.00,0.00}{\textsf{R}_{L}}$  and the input probability distribution $\PP_{\sf{x}}^{\star} \in \Delta(\mc{X})$. For all $\varepsilon>0$, there exists a length $\bar{n}\in \N$ such that for all $n\geq \bar{n}$, there exists a HARQ-code ${c}^{\star}_n \in \mc{C}(n,\textsf{R},L)$ that satisfies equations \eqref{eq:WynerEquivocLz}, for all channel states $(k_1,\ldots, k_L) \in \bigcup_{l = 1}^L \mc{S}_l( \varepsilon,  \textsf{R} , \textcolor[rgb]{0.00,0.00,0.00}{\textsf{R}_{1}}, \ldots , \textcolor[rgb]{0.00,0.00,0.00}{\textsf{R}_{L}} ,  \PP_{\sf{x}}^{\star})$.
\begin{eqnarray}
\mc{P}_{\textsf{e}}\bigg({c}^{\star}_n \; \bigg|k_1,\ldots , k_L\bigg) \leq \varepsilon, \quad 
\mc{L}_{\textsf{e}}\bigg({c}^{\star}_n  \; \bigg|k_1,\ldots , k_L\bigg)  \leq \varepsilon. \label{eq:WynerEquivocLz}
\end{eqnarray}
\end{theorem}
The proof of Theorem \ref{theo:RandomLeszekCode}, stated in Appendix~\ref{sec:proofTheoLz}, involves a new multilevel coding argument that cannot be obtained as a generalization of the coding scheme of  \cite[Appendix A]{TangLiuSpasojevicPoor09}.


In the rest of this article, the optimal sequence of HARQ-codes ${c}^{\star} = ({c}^{\star}_n)_{n\geq 1} $ is called ``Adaptation-Secrecy-Rate-code'' (ASR-code) with parameters  $ \textsf{R},\textcolor[rgb]{0.00,0.00,0.00}{\textsf{R}_{1}}, \ldots, \textcolor[rgb]{0.00,0.00,0.00}{\textsf{R}_{L}}$. The additional degrees of freedom $\textsf{R}_{2}, \ldots ,  \textsf{R}_{L} $ will be exploited to increase the secrecy throughput and to lower the expected number of transmissions and the connection and secrecy outages.


\section{Secrecy Throughput, Connection and Secrecy Outages}\label{sec:ThroughputOutages}

\subsection{Definitions}\label{sec:DefThroughputOutages}

The channels under investigation are controlled by a state parameter $k\in\mc{K}$ observed by the decoder and by the eavesdropper but not by the encoder. We investigate the secure transmission over this state-dependent wiretap channel based on the outage approach. In this setting, the quality of the channel of the eavesdropper is not known by the legitimate encoder and decoder. We introduce the  events $(\mc{A}_l)_{l\in \{1,\ldots, L\}}$ corresponding to the correct decoding \eqref{eq:InformationEventA} and the events $(\mc{B}_l)_{l\in \{1,\ldots, L\}}$ corresponding to the secret transmission  \eqref{eq:InformationEventB}.
\begin{eqnarray}
\mc{A}_l &=&  \bigg\{ \textsf{R}  + \sum_{j\in 1}^l \textcolor[rgb]{0.00,0.00,0.00}{\sf{R_{j}}}  \leq \sum_{j\in 1}^l  I(X_j;Y_j|k_j)  \bigg\}, \label{eq:InformationEventA} \\
\mc{B}_l &=&  \bigg\{\sum_{j\in 1}^l \textcolor[rgb]{0.00,0.00,0.00}{\sf{R_{j}}}  \geq \sum_{j\in 1}^l  I(X_j;Z_j|k_j)  \bigg\} ,\label{eq:InformationEventB}
\end{eqnarray}

\begin{definition}\label{def:OutageEvents} 
The connection outage probability $\mc{P}_{\sf{co}}$ and secrecy outage probability $\mc{P}_{\sf{so}}$ are defined by: 
\begin{eqnarray}
\mc{P}_{\sf{co}} = \mc{P}\bigg(   \bigcap_{l=1}^L \mc{A}_l^c  \bigg), \qquad 
\mc{P}_{\sf{so}}=  \mc{P}\bigg( \bigcup_{l=1}^L \mc{B}_l^c  \bigg) . \label{eq:SecrecyOutageEvent}
\end{eqnarray}
\end{definition}
A connection outage occurs if {for all} transmissions $l\in \{1,\ldots,L\}$, the decoding event $\mc{A}_l$ is not satisfied. A secrecy outage occurs if {there exists} a transmission $l\in \{1,\ldots,L\}$, for which the secrecy event $\mc{B}_l$ is not satisfied.

\begin{remark}
Notation $\mc{A}^c$ stands for the complementary of $\mc{A}$.
Letting the parameters $\textsf{R}_{2}=  \ldots = \textcolor[rgb]{0.00,0.00,0.00}{\textsf{R}_{L}} = 0 $, this implies that $\mc{A}_{l-1}\subset \mc{A}_l$,  $\mc{B}_l\subset \mc{B}_{l-1}$ and the definitions of $\mc{P}_{\sf{co}}$ and $\mc{P}_{\sf{so}}$ reduce to those shown in  \cite[Eqs.~(21), (22)]{TangLiuSpasojevicPoor09}. 
\end{remark}

\begin{proposition}\label{prop:SecrecyOutage}
Suppose that the random events  $(\mc{B}_l)_{l\in \{1,\ldots, L\}}$ are independent of the random events $(\mc{A}_l)_{l\in \{1,\ldots, L\}}$. The secrecy outage probability writes:
\begin{small}
\begin{eqnarray}
 \mc{P}_{\sf{so}} &=&  1 -  \sum_{j=2}^{L-1}    \mc{P}\bigg(   \bigcap_{i = 1}^j  \mc{B}_i  \bigg)  \cdot \Bigg(  \mc{P} \bigg(   \bigcap_{i = 1}^{j-1}  \mc{A}^c_i  \bigg) -  \mc{P}\bigg(  \bigcap_{i = 1}^{j}  \mc{A}^c_i   \bigg) \Bigg)\nonumber \\ 
 &-&\mc{P}\bigg(   \mc{B}_1  \bigg)  \cdot  \mc{P} \bigg(     \mc{A}_1  \bigg) -  \mc{P}\bigg(   \bigcap_{i = 1}^L  \mc{B}_i  \bigg)  \cdot  \mc{P} \bigg(   \bigcap_{i = 1}^{L-1}  \mc{A}^c_i  \bigg) .
    \end{eqnarray} 
  \end{small}
 \end{proposition}
Proof of Prop. \ref{prop:SecrecyOutage} is stated in App. \ref{ref:ProofProp}. This formulation will be used for
discrete channels in Sec. \ref{sec:DiscreteChannels} and Gaussian channel in Sec. \ref{sec:RayleighGaussian}.
We denote by $ \textsf{L} \in \{1,\ldots,L\}$, the random number of transmissions that depends on channel states parameters $(k_1,\ldots, k_L)$ and rate parameters $(  \textsf{R}, \textcolor[rgb]{0.00,0.00,0.00}{\textsf{R}_{1}}, \ldots,  \textcolor[rgb]{0.00,0.00,0.00}{\textsf{R}_{L}})$. 
\begin{eqnarray}
\PP(\textsf{L} = 1 )  &=&  \mc{P} \bigg(   \mc{A}_1 \bigg),  \label{eq:ProbaTransmission1} \\
\PP(\textsf{L} = l )  &=&    \mc{P} \bigg(   \bigcap_{j = 1}^{l-1}  \mc{A}^c_j \cap \mc{A}_l \bigg)  , \qquad  \forall l \in \{2, \ldots, L-1\} \nonumber \\
&=&    \mc{P} \bigg(   \bigcap_{j = 1}^{l-1}  \mc{A}^c_j  \bigg) -  \mc{P}\bigg(  \bigcap_{j = 1}^{l}  \mc{A}^c_j   \bigg) ,  \label{eq:ProbaTransmission} \\
\PP(\textsf{L} = L )  &=&  \mc{P} \bigg(   \bigcap_{j = 1}^{L-1}  \mc{A}^c_j \bigg).  \label{eq:ProbaTransmissionL} 
\end{eqnarray}
The expected number of transmissions $\E\big[\textsf{L} \big]$ is given by:
\begin{eqnarray}
\E\big[\textsf{L} \big]  &=&  \sum_{l=1}^L  l \cdot   \PP(\textsf{L} = l )  = 1 +   \sum_{l=1}^{L-1} \mc{P}\bigg(  \bigcap_{j = 1}^{l}  \mc{A}^c_j   \bigg). \label{eq:ExpectedNbTransmission}
\end{eqnarray}

Since the number of transmissions $\textsf{L}$ is a random variable, the expected number of bits correctly decoded is given by the Renewal-Reward Theorem as in \cite{ZorziRao96}, \cite{CaireTuninettiHARQ2001}. 

\begin{definition}\label{def:SecrecyThroughput}
The secrecy throughput $\eta$ is defined as the expected number of bits correctly decoded by the legitimate decoder per channel use and can be obtained from the renewal-reward approach 
\begin{eqnarray}
\textcolor[rgb]{0.00,0.00,0.00}{\eta}&=&\max_{ \textsf{R}, \atop \textcolor[rgb]{0.00,0.00,0.00}{\textsf{R}_{1}}, \ldots \textcolor[rgb]{0.00,0.00,0.00}{\textsf{R}_{L}}, } \frac{\E[\textsf{R}]}{\E[\textsf{L}]}
= \max_{  \textsf{R}, \atop\textcolor[rgb]{0.00,0.00,0.00}{\textsf{R}_{1}}, \ldots \textcolor[rgb]{0.00,0.00,0.00}{\textsf{R}_{L}}, } \frac{\textsf{R} \cdot (1 - \mc{P}_{\sf{co}})}{ 1 +   \sum_{l=1}^{L-1} \mc{P}\big(  \bigcap_{j = 1}^{l}  \mc{A}^c_j   \big)  }, \nonumber \\ 
&\text{u.c.}&
\begin{cases}
\mc{P}_{\sf{co}} & \leq \xi_{\sf{c}}, \\
\mc{P}_{\sf{so}}  &\leq \xi_{\sf{s}}. \label{eq:ThroughputDefinition}\\ 
\end{cases}
\end{eqnarray}
\end{definition}
The maximum is taken over the parameters $ \textsf{R}, \textsf{R}_{1}, \ldots, \textsf{R}_{L}$, such that the connection outage probability and the secrecy outage probability are lower than $\xi_{\sf{c}}$ and $\xi_{\sf{s}}$,  which are the constraints defined according to  the requirements on the secrecy and reliability.



\subsection{Example: Discrete Channel States}\label{sec:DiscreteChannels}

To illustrate the definitions we introduced, we consider the  scenario represented by Fig. \ref{fig:IndependentChannelb}, in which the channel states of the legitimate decoder and of the eavesdropper are binary and uniformly distributed over $\{k^y,k'^y\}$ and $\{k^z,k'^z\}$. We define the operating point as $\xi_{\sf{c}}= 0.25$ and $ \xi_{\sf{s}}= 0.125$, and assume the maximum number of transmissions is $L=2$.
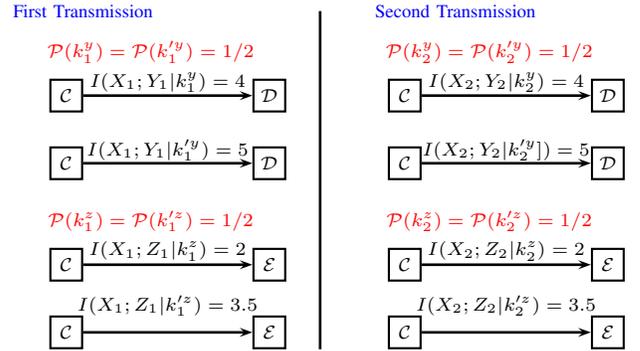
\begin{figure}[!h]
\begin{scriptsize}
\begin{center}
\psset{xunit=0.45cm,yunit=0.45cm}
\begin{pspicture}(0,-2.5)(16,8)
\rput[u](1,8){\textcolor[rgb]{0.0,0.0,1.0}{First Transmission}}
\psframe(0,0)(1,1)
\psframe(0,-2)(1,-1)
\psframe(6,0)(7,1)
\psframe(6,-2)(7,-1)
\psline[linewidth=1pt]{->}(1,0.5)(6,0.5)
\psline[linewidth=1pt]{->}(1,-1.5)(6,-1.5)
\rput[u](3.5,0.9){$I (X_1; Z_1|{k}^z_1) = 2$}
\rput[u](3.5,-0.7){$I (X_1; Z_1|{k}'^z_1) = 3.5$}
\rput(0.5,0.5){$\C$}
\rput(0.5,-1.5){$\C$}
\rput(6.5,0.5){$\mc{E}$}
\rput(6.5,-1.5){$\mc{E}$}
\rput(3,1.8){\textcolor[rgb]{1.0,0.0,0.0}{$\mc{P}(k^z_1) = \mc{P}(k'^z_1) = 1/2$}}
\psframe(0,5)(1,6)
\psframe(0,3)(1,4)
\psframe(6,5)(7,6)
\psframe(6,3)(7,4)
\psline[linewidth=1pt]{->}(1,5.5)(6,5.5)
\psline[linewidth=1pt]{->}(1,3.5)(6,3.5)
\rput[u](3.5,5.9){$I (X_1; Y_1|k^y_1) = 4$}
\rput[u](3.5,3.9){$I (X_1; Y_1|k'^y_1) = 5$}
\rput(0.5,5.5){$\C$}
\rput(0.5,3.5){$\C$}
\rput(6.5,5.5){$\D$}
\rput(6.5,3.5){$\D$}
\rput(3,6.8){\textcolor[rgb]{1.0,0.0,0.0}{$\mc{P}(k^y_1) = \mc{P}(k'^y_1) = 1/2$}}
\rput[u](12,8){\textcolor[rgb]{0.0,0.0,1.0}{Second Transmission}}
\psline[linewidth=1pt]{-}(8,-2)(8,8)
\psframe(10,0)(11,1)
\psframe(10,-2)(11,-1)
\psframe(16,0)(17,1)
\psframe(16,-2)(17,-1)
\psline[linewidth=1pt]{->}(11,0.5)(16,0.5)
\psline[linewidth=1pt]{->}(11,-1.5)(16,-1.5)
\rput[u](13.5,0.9){$I (X_2; Z_2|{k}^z_2) = 2$}
\rput[u](13.5,-0.7){$I (X_2; Z_2|{k}'^z_2) = 3.5$}
\rput(10.5,0.5){$\C$}
\rput(10.5,-1.5){$\C$}
\rput(16.5,0.5){$\mc{E}$}
\rput(16.5,-1.5){$\mc{E}$}
\rput(13,1.8){\textcolor[rgb]{1.0,0.0,0.0}{$\mc{P}(k^z_2) = \mc{P}(k'^z_2) = 1/2$}}
\psframe(10,5)(11,6)
\psframe(10,3)(11,4)
\psframe(16,5)(17,6)
\psframe(16,3)(17,4)
\psline[linewidth=1pt]{->}(11,5.5)(16,5.5)
\psline[linewidth=1pt]{->}(11,3.5)(16,3.5)
\rput[u](13.5,5.9){$I (X_2; Y_2|k^y_2) = 4$}
\rput[u](13.5,3.9){$I (X_2; Y_2|k'^y_{2}]) = 5$}
\rput(10.5,5.5){$\C$}
\rput(10.5,3.5){$\C$}
\rput(16.5,5.5){$\D$}
\rput(16.5,3.5){$\D$}
\rput(13,6.8){\textcolor[rgb]{1.0,0.0,0.0}{$\mc{P}(k^y_2) = \mc{P}(k'^y_2) = 1/2$}}
\end{pspicture}
\caption{In both transmissions, the capacity of the channel to the legitimate decoder takes two possible values $\{4 , 5 \}$ with probability $(1/2, 1/2)$ and the capacity of the channel to the eavesdropper takes two possible values $ \{2 , 3.5 \}$ with probability $(1/2, 1/2)$.}\label{fig:IndependentChannelb}
\end{center}
\end{scriptsize}
\end{figure}
We investigate the secrecy throughput of the ASR-code whose existence is stated in Theorem  \ref{theo:RandomLeszekCode} and we compare its performance to the protocols shown in \cite{TangLiuSpasojevicPoor09}  and in \cite{TomasinLaurenti14}.



$\bullet$ The secure HARQ protocol of \cite{TangLiuSpasojevicPoor09} is a particular case of the ASR-code in which the dummy-message rate $ \textsf{R}_{2} = 0$ is zero. As depicted on fig. \ref{fig:RateRegionsPoor}, after $L=2$ transmissions, the decoding is correct if:
\begin{eqnarray}
\textsf{R} + \textsf{R}_1 &\leq&  I(X_1;Y_1|k_1) +  I(X_2;Y_2|k_2),\label{eq:TangD}
\end{eqnarray}
and the transmission is secret if:
\begin{eqnarray}
\textsf{R}_1 &\geq&  I(X_1;Y_1|k_1) +  I(X_2;Y_2|k_2).\label{eq:TangS}
\end{eqnarray}
\begin{figure}[!ht]
\begin{center}
\psset{xunit=1cm,yunit=1cm}
\begin{pspicture}(-1,-1)(5,5)
\pspolygon*[linecolor=vertclair](0,3)(3,0)(5.5,0)(5.5,4.5)(0,4.5) 
\pspolygon*[linecolor=red](0,0)(0,1.5)(1.5,0)
\pspolygon[fillstyle=hlines,hatchcolor=blue,linecolor=white](3,0)(5.5,0)(5.5,4.85)(3,4.85)
\psline[linewidth=0.5pt]{->}(-1,0)(6,0)
\psline[linewidth=0.5pt]{->}(0,-1)(0,5)
\rput[d](5.5,-0.3){$I(X_1;Y_1|k_1)$}
\rput[d](5.5,-0.7){$I(X_1;Z_1|k_1)$}
\rput[r](-0.2,4.8){$I(X_2;Y_2|k_2)$}
\rput[r](-0.2,4.4){$I(X_2;Z_2|k_2)$}
\psline[linewidth=0.5pt,linestyle=dotted](0,3)(3,0)
\psdots(1.5,0)(3,0)(0,3)(0,1.5)
\rput[d](1.5,-0.3){$\textsf{R}_{1}$}
\rput[u](2,3.5){(D)}
\rput[u](0.5,0.5){(S)}
\rput[r](-0.2,1.5){$\textsf{R}_{1}$}
\rput[r](-0.2,3){$\textsf{R} + \textsf{R}_{1}$}
\psline[linewidth=0.5pt,linestyle=dotted](1.5,0)(0,1.5)
\psline[linewidth=1pt,linestyle=dashed,linecolor = blue](3,0)(3,5)
\rput[d](3,-0.3){ $\; \textsf{R} + \textsf{R}_{1}$}
\rput[l](3,5){ $\;\textcolor[rgb]{0.00,0.00,1.00}{\textsf{Ack}_1}$}
\rput[r](3,5){$ \textcolor[rgb]{0.00,0.00,1.00}{\textsf{Nack}_1} \;$ }
\end{pspicture}
\caption{Regions of correct decoding (D) and secret transmission (S) of \cite{TangLiuSpasojevicPoor09}, corresponding to equations \eqref{eq:TangD} and \eqref{eq:TangS}.}\label{fig:RateRegionsPoor}
\end{center}
\end{figure}
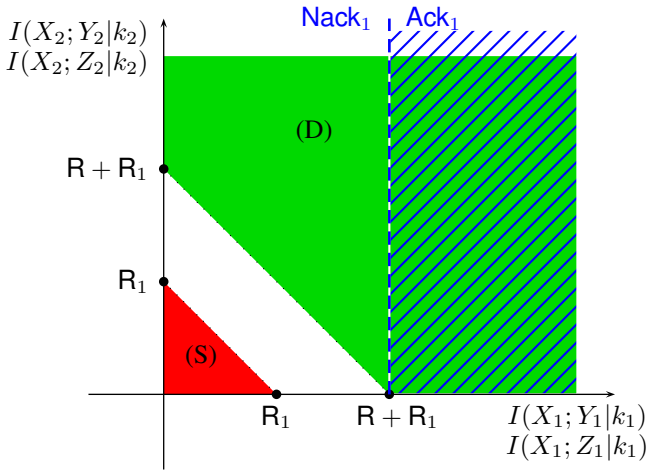

$\bullet$ The S-HARQ protocol of \cite[Sec. V]{TomasinLaurenti14} involves multiple dummy-message rates $(\textsf{R}_1 , \textsf{R}_2)$. As depicted on fig. \ref{fig:RateRegionsTomasin}, after $L=2$ transmissions, the decoding is correct if:
\begin{eqnarray}
\textsf{R} &\leq&  \max\Big(I(X_1;Y_1|k_1) -  \textsf{R}_1 , 0 \Big)  \nonumber \\
&+&   \max\Big(I(X_2;Y_2|k_2)  -  \textsf{R}_2 , 0\Big), \label{eq:TomasinD}
\end{eqnarray}
and the transmission is secret if:
\begin{eqnarray}
\textsf{R}_1 \geq  I(X_1;Y_1|k_1) ,\qquad
\textsf{R}_2 \geq   I(X_2;Y_2|k_2).\label{eq:TomasinS}
\end{eqnarray}

Fig. \ref{fig:RateRegions}, \ref{fig:RateRegionsPoor} and \ref{fig:RateRegionsTomasin} show that the decoding and the secrecy regions are different for the ASR-code and for the protocols of \cite{TangLiuSpasojevicPoor09} and \cite{TomasinLaurenti14}. 
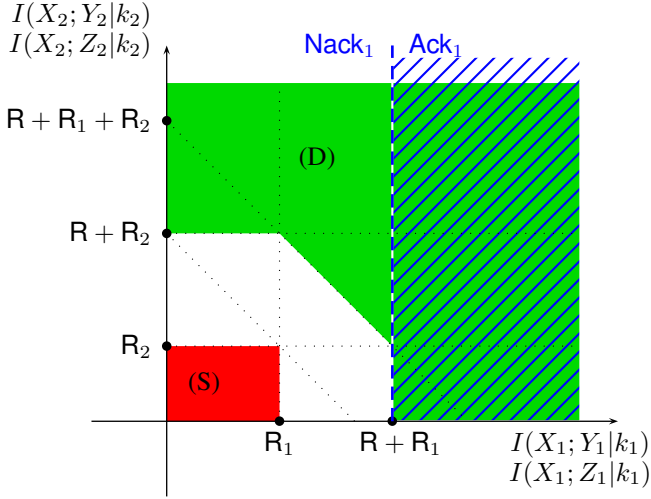
\begin{figure}[!ht]
\begin{center}
\psset{xunit=1cm,yunit=1cm}
\begin{pspicture}(-1,-1)(5,5.5)
\pspolygon*[linecolor=vertclair](0,2.5)(1.5,2.5)(3,1)(3,0)(5.5,0)(5.5,4.5)(0,4.5) 
\pspolygon*[linecolor=red](0,0)(0,1)(1.5,1)(1.5,0)
\pspolygon[fillstyle=hlines,hatchcolor=blue,linecolor=white](3,0)(5.5,0)(5.5,4.85)(3,4.85)
\psline[linewidth=0.5pt]{->}(-1,0)(6,0)
\psline[linewidth=0.5pt]{->}(0,-1)(0,5)
\rput[d](5.5,-0.3){$I(X_1;Y_1|k_1)$}
\rput[d](5.5,-0.7){$I(X_1;Z_1|k_1)$}
\rput[r](-0.2,5.4){$I(X_2;Y_2|k_2)$}
\rput[r](-0.2,5){$I(X_2;Z_2|k_2)$}
\psline[linewidth=0.5pt,linestyle=dotted](0,4)(4,0)
\psline[linewidth=0.5pt,linestyle=dotted](0,2.5)(5.5,2.5)
\psdots(1.5,0)(0,4)(0,1)(0,2.5)(3,0)
\rput[d](1.5,-0.3){$\textsf{R}_{1}$}
\rput[r](-0.2,1){$\textsf{R}_{2}$}
\rput[r](-0.2,2.5){$\textsf{R} + \textsf{R}_{2}$}
\rput[r](-0.2,4){$\textsf{R} + \textsf{R}_{1} + \textsf{R}_{2}$}
\psline[linewidth=0.5pt,linestyle=dotted](0,1)(5.5,1)
\psline[linewidth=0.5pt,linestyle=dotted](1.5,0)(1.5,4.5)
\psline[linewidth=0.5pt,linestyle=dotted](2.5,0)(0,2.5)
\psline[linewidth=1pt,linestyle=dashed,linecolor = blue](3,0)(3,5)
\rput[d](3,-0.3){ $\; \textsf{R} + \textsf{R}_{1}$}
\rput[l](3,5){ $\;\textcolor[rgb]{0.00,0.00,1.00}{\textsf{Ack}_1}$}
\rput[r](3,5){$ \textcolor[rgb]{0.00,0.00,1.00}{\textsf{Nack}_1} \;$ }
\rput[u](2,3.5){(D)}
\rput[u](0.5,0.5){(S)}
\end{pspicture}
\caption{Regions of correct decoding (D) and secret transmission (S) of \cite{TomasinLaurenti14}, corresponding to equations \eqref{eq:TomasinD} and \eqref{eq:TomasinS}.}
\label{fig:RateRegionsTomasin}
\end{center}
\end{figure}


\begin{figure}[ht!]
\begin{scriptsize}
\psfragscanon
\psfrag{labelx}[][]{Secrecy Rate: $\textsf{R}$}
\psfrag{labely}[][]{Secrecy Throughput: $\eta$}
\psfrag{datadatadatadatadata1}[l][l]{ASR-code}
\psfrag{datadatadatadatadata2}[l][l]{Protocol of \cite{TangLiuSpasojevicPoor09}}
\psfrag{datadatadatadatadata3}[l][l]{Protocol of  \cite{TomasinLaurenti14}}
\includegraphics[width=0.5\textwidth]{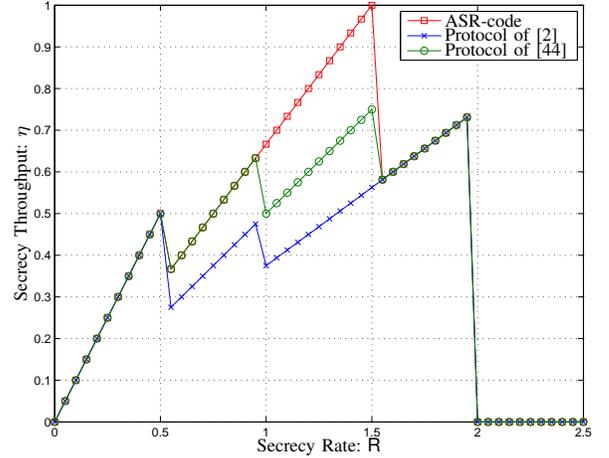}
\end{scriptsize}
\caption{Secrecy throughput for $L=2$ possible transmissions and the constraints $\xi_{\sf{c}}= 0.25$, $ \xi_{\sf{s}}= 0.125$.}\label{fig:DiscreteChannelsThroughput_2017_05_16}
\end{figure}

Fig. \ref{fig:DiscreteChannelsThroughput_2017_05_16} compares the secrecy throughput of ASR-code, and of the protocols of \cite{TangLiuSpasojevicPoor09} and \cite{TomasinLaurenti14}; we observe the following
\begin{itemize}
\item The ASR-code outperforms both protocols \cite{TangLiuSpasojevicPoor09} and  \cite{TomasinLaurenti14} for the secrecy rate $\textsf{R}= 1.5$, dummy message rates $(\textsf{R}_1, \textsf{R}_2 ) = (3.5,2)$ and outage probabilities $(\mc{P}_{\sf{co}},\mc{P}_{\sf{so}} )= (0,0.125)$. 
\item For the protocol of \cite{TangLiuSpasojevicPoor09}, the optimal dummy-message rate $\textsf{R}_1^{\text{\cite{TangLiuSpasojevicPoor09}}} = 7$ corresponding to $L=2$ transmissions is high and prevents the first transmission to be decoded correctly.
\item 
For  the protocol of \cite{TomasinLaurenti14}, the optimal second dummy-message rate $\textsf{R}_2^{\text{\cite{TomasinLaurenti14}}} = 3.5$ is higher than $\textsf{R}_2 =2$ for ASR-code. Hence, when the secrecy rate exceeds $\textsf{R}\geq 1.5$, the connection outage probability increases to $\mc{P}_{\sf{co}}=0.25$ and reduces the secrecy throughput.
\item
In the example we show, the ASR-code provides more than $33\%$ of increase compared to the protocols of \cite{TangLiuSpasojevicPoor09} and  \cite{TomasinLaurenti14}. However, we hasten to say that the improvement depends on the adopted distribution of the channel gains. In particular, if the values of the channels to the eavesdropper are replaced by $ \{2 , 3 \}$ (instead of  $ \{2 , 3.5 \}$), the protocol of \cite{TomasinLaurenti14} provides the same secrecy throughput $\eta = 1.333$ as the ASR-code, whereas the protocol of \cite{TangLiuSpasojevicPoor09} provides a lower secrecy throughput of $\eta = 1.125$. Therefore, while we are sure our approach outperforms \cite{TangLiuSpasojevicPoor09}, the direct comparison with \cite{TomasinLaurenti14} is not obvious, as also noted in  \cite[pp.1714]{TomasinLaurenti14}. Despite this cautionary statement, we did not find any example where the throughput of \cite{TomasinLaurenti14} is larger than the one offered by the ASR-code we propose.
\end{itemize}

 \section{Rayleigh Block-Fading Gaussian Wiretap Channels}\label{sec:RayleighGaussian}
 
   \subsection{Channel Model}\label{sec:ChannelModel}

We consider the Gaussian wiretap channel with Rayleigh block-fading represented in Fig.~\ref{fig:GaussianStateDepParallelWiretapsState} and defined as
\begin{eqnarray}
Y = h_{\textsf{d}} \cdot X + N_{\textsf{d}}, \qquad 
Z = h_{\textsf{e}} \cdot X + N_{\textsf{e}}.\label{eq:GaussianChannel2}
\end{eqnarray}
where $N_{\textsf{d}}$ and $N_{\textsf{e}}$ are i.i.d. zero-mean, unit-variance Gaussian variables. In this work, we consider the canonical model of independent block-fading channels. The applicability of the ASR-code in the case of correlated fading \cite{HarsiniLahoutiLevoratoZorzi11,KimChoiBanSung11,ChaitanyaLarsson14,ShiDingMaTam15} requires further study but goes beyond the scope of this work which focuses on the code design problem.

We assume a normalized power constraint on the channel input $\E\big[ |X|^2 \big] \leq P = 1$.
\begin{figure}[!ht]
\begin{center}
\psset{xunit=1cm,yunit=1cm}
\begin{pspicture}(-0.75,-1.3)(7.5,1.5)
\psline[linewidth=1pt]{->}(-1,0.5)(0,0.5)
\rput[u](-0.5,0.7){$M$}
\psframe(0,0)(1,1)
\rput[u](0.5,0.5){$\mc{C}$}
\psframe(6,0)(7,1)
\rput[u](6.5,0.5){$\mc{D}$}
\psline[linewidth=1pt]{->}(7,0.5)(8,0.5)
\rput[u](7.5,0.75){$\hat{M}$}
\psframe(6,-1.5)(7,-0.5)
\rput[u](6.5,-1){$\mc{E}$}
\rput[u](1.5,0.8){$X$}
\rput[u](5.5,0.8){$Y$}
\rput[u](5.5,-0.7){$Z$}
\psline[linewidth=1pt]{->}(1,0.5)(2.75,0.5)
\psline[linewidth=1pt]{->}(3.25,0.5)(4.25,0.5)
\psline[linewidth=1pt]{->}(4.25,0.5)(6,0.5)
\pscircle(3,0.5){0.25}
\pscircle(4.5,0.5){0.25}
\psline[linewidth=1pt](4.5,0.25)(4.5,0.75)
\psline[linewidth=1pt](2.823,0.676)(3.176,0.324)
\psline[linewidth=1pt](2.823,0.324)(3.176,0.676 )
\psline[linewidth=1pt]{->}(2,0.5)(2,-1)(2.75,-1)
\psline[linewidth=1pt]{->}(3.25,-1)(4.25,-1)
\psline[linewidth=1pt]{->}(4.25,-1)(6,-1)
\pscircle(3,-1){0.25}
\pscircle(4.5,-1){0.25}
\psline[linewidth=1pt](4.5,-1.25)(4.5,-0.75)
\psline[linewidth=1pt](2.823,-1.176)(3.176,- 0.823)
\psline[linewidth=1pt](2.823,- 0.823)(3.176,-1.176 )
\psline[linewidth=1pt]{->}(3,1.2)(3,0.75)
\psline[linewidth=1pt]{->}(4.5,1.2)(4.5,0.75)
\rput[u](3,1.35){$h_{\textsf{d}}$}
\rput[u](4.5,1.35){$N_{\textsf{d}}$}
\psline[linewidth=1pt]{->}(3,-0.3)(3,-0.75)
\psline[linewidth=1pt]{->}(4.5,-0.3)(4.5,-0.75)
\rput[u](3,-0.15){$h_{\textsf{e}}$}
\rput[u](4.5,-0.15){$N_{\textsf{e}}$}
\end{pspicture}
\caption{Gaussian wiretap channel with  Rayleigh block-fading $(h_{\textsf{d}}, h_{\textsf{e}})$. }
\label{fig:GaussianStateDepParallelWiretapsState}
\end{center}
\end{figure}
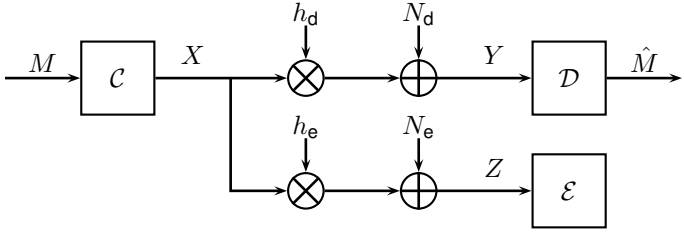
The state parameters $k = (h_{\textsf{d}}, h_{\textsf{e}}) \in \mc{K}$ are fading coefficients, distributed i.i.d. from one block to another with Rayleigh probability distribution. 
Since the mean of noise and power are normalized to 1, we introduce the notation $\textsf{SNR}_{\textsf{d}} = |h_{\textsf{d}}|^2$ and $\textsf{SNR}_{\textsf{e}} = |h_{\textsf{e}}|^2$. The mean $\textsf{SNR}$s  are denoted by $\gamma_{\textsf{d}} = \E[\textsf{SNR}_{\textsf{d}} ] = \E\big[|h_{\textsf{d}}|^2 \big]$ and $\gamma_{\textsf{e}} = \E[\textsf{SNR}_{\textsf{e}} ] = \E\big[|h_{\textsf{e}}|^2 \big]$. For $x\geq 0$, the probability density function $f(x)$ and the cumulative distribution function $F(x)$ of the  \textsf{SNR}s are defined by
   \begin{eqnarray}
 f(x) =  \frac{1}{\gamma}\cdot e^{-\frac{x}{\gamma}}, \qquad F(x) =    1 -  e^{-\frac{x}{\gamma}} . \label{eq:CDF}
\end{eqnarray}
so the mutual informations write as
\begin{eqnarray}
I(X;Y|h_{\textsf{d}}) &=& \log( 1 + \textsf{SNR}_{\textsf{d}}),\label{eq:MutualInfo1}\\
I(X;Z|h_{\textsf{e}}) &=& \log( 1 + \textsf{SNR}_{\textsf{e}}) . \label{eq:MutualInfo2}
\end{eqnarray}
and depend on the random fading coefficients $k = (h_{\textsf{d}}, h_{\textsf{e}}) \in \mc{K}$.
   
The constraints $\xi_{\sf{c}}$ and $\xi_{\sf{s}}$ are not always compatible since the outage constraints $\mc{P}_{\sf{co}} \leq\xi_{\sf{c}} $ and $\mc{P}_{\sf{so}} \leq\xi_{\sf{s}} $ may not be satisfied simultaneously. We characterize the trade-off between connection outage probability and secrecy outage probability when only one transmission is allowed, i.e., $L=1$.

\begin{figure}[ht!]
\begin{scriptsize}
\psfragscanon
\psfrag{labelx}[][l]{Transmission Rate: $\textsf{R}$}
\psfrag{labely}[][l]{Secrecy Throughput}
\psfrag{aa}[l][B]{$\gamma_d = 15 [dB]$}
\psfrag{bb}[l][B]{$\gamma_e = 5 [dB]$}
\psfrag{cc}[l][B]{$\textsf{L}=1$}
\psfrag{datadatadata1}[l][l]{ $\xi_s = 10^{-2}$}
\psfrag{datadatadata2}[l][l]{ $\xi_s = 10^{-4}$}
\psfrag{datadatadata3}[l][l]{ $\xi_s = 10^{-6}$}
\psfrag{max2}[l][B]{maximal $\textsf{R}$ for $\xi_c = 0.75$ and $\xi_s = 10^{-2}$}
\psfrag{max4}[l][B]{maximal $\textsf{R}$ for $\xi_c = 0.75$ and $\xi_s = 10^{-4}$}
\psfrag{max6}[l][B]{maximal $\textsf{R}$ for $\xi_c = 0.75$ and $\xi_s = 10^{-6}$}
\psfrag{opt2}[l][B]{Optimal $\textsf{R}$ for $\xi_s = 10^{-2}$}
\psfrag{opt4}[l][B]{Optimal $\textsf{R}$ for $\xi_s = 10^{-4}$}
\psfrag{opt6}[l][B]{Optimal $\textsf{R}$ for $\xi_s = 10^{-6}$}
\includegraphics[width=0.5\textwidth]{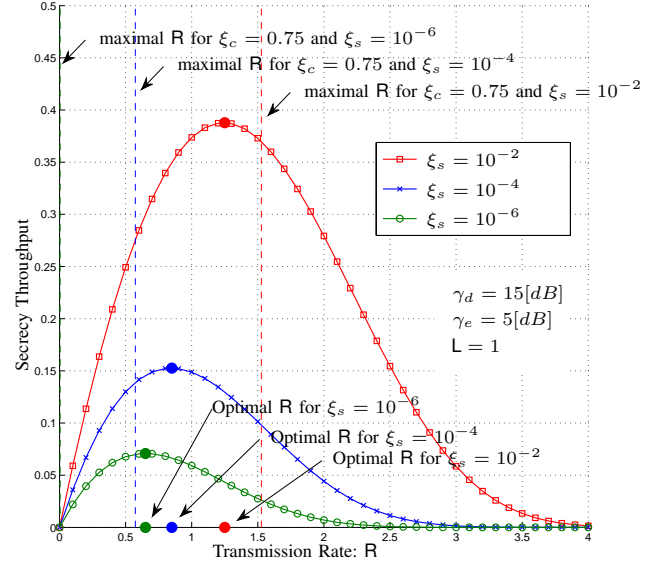}
\end{scriptsize}
\vspace{-0.9cm}
\caption{Secrecy throughput depending on the secrecy rate  $\textsf{R}\geq0$, for different secrecy constraints $\xi_{\sf{s}} \in \big\{10^{-2},10^{-4},10^{-6}\big\}$ and a single $L=1$ transmission. Vertical dashed lines represents the maximal secrecy rate $\textsf{R}$ corresponding to the constraint $\xi_{\sf{c}} = 0.75$.}\label{fig:Throughput1T}
\end{figure}

\begin{theorem}\label{theo:FeasibleOutOneTrans} Consider the case of $L=1$ transmission.\\ 
$\bullet$ The constraints $\xi_{\sf{c}}$ and $\xi_{\sf{s}}$ are compatible 
if and only if
\begin{eqnarray}
 \xi_{\sf{s}} \geq  \bigg(1 - \xi_{\sf{c}}\bigg)^{\frac{\gamma_{\textsf{d}}}{ \gamma_{\textsf{e}}}} \quad \Longleftrightarrow \quad  \bigg(\xi_{\sf{s}}\bigg)^{\gamma_{\textsf{e}}} -  \bigg(1 - \xi_{\sf{c}}\bigg)^{\gamma_{\textsf{d}}}\geq 0 .\label{eq:theoTradeOff1}
\end{eqnarray}
$\bullet$ For a fixed secrecy rate $\textsf{R}\geq 0$, the contraints $\xi_{\sf{c}}$ and $\xi_{\sf{s}}$ are compatible if and only if
\begin{eqnarray}
\textcolor[rgb]{0,0,0}{\textsf{R}} & \leq&  \log_2 \bigg( \frac{1 -  \gamma_{\textsf{d}} \cdot \ln(1 - \xi_{\sf{c}})}{1 - \gamma_{\textsf{e}} \cdot \ln(\xi_{\sf{s}})}\bigg)   . \label{eq:theoTradeOff2}
\end{eqnarray}
\end{theorem}
The proof of Theorem \ref{theo:FeasibleOutOneTrans} is stated in App. \ref{sec:FeasibleOutOneTrans}. Equation \eqref{eq:theoTradeOff1} emphasizes that the trade-off between the connection and the secrecy outage probability only depends on the ratio $\gamma_{\textsf{d}} / \gamma_{\textsf{e}}$, \ie  the difference $\gamma_{\textsf{d}} - \gamma_{\textsf{e}}$ in [dB]. Fig. \ref{fig:Throughput1T} represents the secrecy throughput for $L=1$ transmission depending on the rate parameter $\textsf{R}$, for different constraints $(\xi_{\sf{c}},\xi_{\sf{s}})$. The shape of the curve  depends on the secrecy outage constraint $\mc{P}_{\sf{so}} \leq \xi_{\sf{s}}$. The connection outage constraint $\mc{P}_{\sf{co}} \leq \xi_{\sf{c}}$ truncates the secrecy throughput at the dashed lines.


\begin{figure}[ht!]
\begin{scriptsize}
\psfragscanon
\psfrag{aa}[l][B]{$\gamma_d = 15 [dB]$}
\psfrag{bb}[l][B]{$\gamma_e = 5 [dB]$}
\psfrag{cc}[l][B]{$\textsf{R}=0$}
\psfrag{l1}[rb][]{$\textsf{L}=1$ transmissions}
\psfrag{l2}[rb][]{$\textsf{L}=2$ transmissions}
\psfrag{l4}[rb][]{$\textsf{L}=4$ transmissions}
\psfrag{l8}[rb][]{$\textsf{L}=8$ transmissions}
\psfrag{labelx}[][]{Connection Outage Probability: $\mathcal{P}_{\textsf{co}}$}
\psfrag{labely}[][]{Secrecy Outage Probability: $\mathcal{P}_{\textsf{so}}$}
\psfrag{datadatadatadatadatadatadata1}[l][l]{{Trade-off for $\textsf{L}=1$ transmission}}
\psfrag{datadatadatadatadatadatadata2}[l][l]{{Trade-off for ASR-code}}
\psfrag{datadatadatadatadatadatadata3}[l][l]{{Trade-off for the protocol of \cite{TangLiuSpasojevicPoor09}}}
\includegraphics[width=0.5\textwidth]{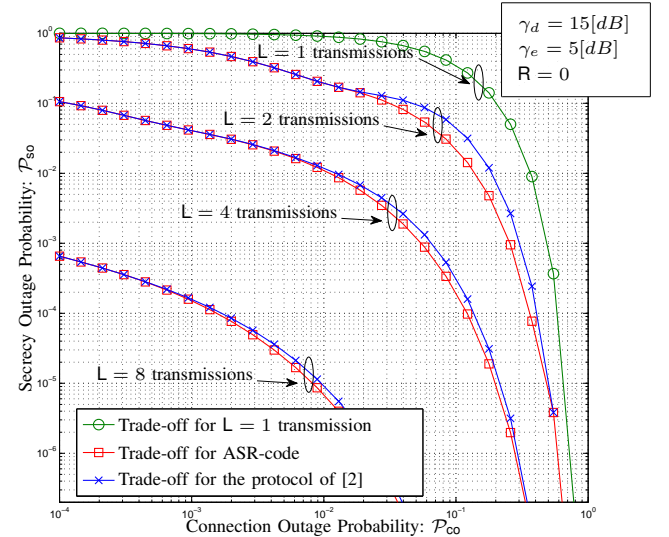}
\end{scriptsize}
\vspace{-0.9cm}
\caption{Trade-off between the connection $\mc{P}_{\sf{co}}$ and secrecy $\mc{P}_{\sf{so}}$ outage probability, for zero rate $\textsf{R} =0$ and number of transmissions $L\in \{1,2,4,8\}$. }\label{fig:PlotTradeoffOutagePsoPco_2016_04_15}
\end{figure}


  \subsection{Multiple Transmissions}\label{sec:OutageMultipleTrans}
  
We propose a numerical optimization of the secrecy throughput with respect to the rate parameters for the case of $L\geq2 $ multiple transmissions. 

Since our objective is to demonstrate that the  ASR-code outperforms the HARQ-code of \cite{TangLiuSpasojevicPoor09}, we show a simple example where the dummy-message rate parameters $\textsf{R}_2 = \textsf{R}_3 = \ldots = \textsf{R}_{L}$ are equal after the second transmission. This makes the presentation easier and avoids the tedious optimization which depends only on three parameters: $(\textsf{R},\textsf{R}_1,\textsf{R}_2)$. 

For Rayleigh fading channels the protocol of \cite{TangLiuSpasojevicPoor09} outperforms the one proposed in \cite{TomasinLaurenti14} as can be seen in \cite[Fig. 6, Fig. 7, Fig. 8]{TomasinLaurenti14}. Thus, we only need to compare the performance of the ASR-code we proposed with the protocol of \cite{TangLiuSpasojevicPoor09}. The main difference is that the latter uses two parameters $(\textsf{R},\textsf{R}_1)$, while the ASR-code uses three rates $(\textsf{R},\textsf{R}_1,\textsf{R}_2)$.

    \begin{figure}[ht!]
        \begin{center}
   \begin{scriptsize}
\psfragscanon
\psfrag{aa}[l][B]{$\gamma_d = 15 [dB]$}
\psfrag{bb}[l][B]{$\gamma_e = 5 [dB]$}
\psfrag{cc}[l][B]{$\textsf{L}=8$}
\psfrag{s1}[r][]{$\xi_s = 10^{-2}$}
\psfrag{c1}[r][]{$\xi_c = 1 $}
\psfrag{s2}[r][]{$\xi_s = 10^{-4}$}
\psfrag{c2}[r][]{$\xi_c = 1 $}
\psfrag{s3}[r][]{$\xi_s = 10^{-6}$}
\psfrag{c3}[r][]{$\xi_c = 1 $}
\psfrag{s4}[r][]{$\xi_s = 10^{-2}$}
\psfrag{c4}[r][]{$\xi_c = 10^{-2} $}
\psfrag{labelx}[][]{Secrecy Rate: $\textsf{R}$}
\psfrag{labely}[][]{Secrecy Throughput}
\psfrag{datadatadatadatadata1}[l][l]{{Performance of the ASR-code}}
\psfrag{datadatadatadatadata2}[l][l]{{Protocol of \cite{TangLiuSpasojevicPoor09}}}
\includegraphics[width=0.5\textwidth]{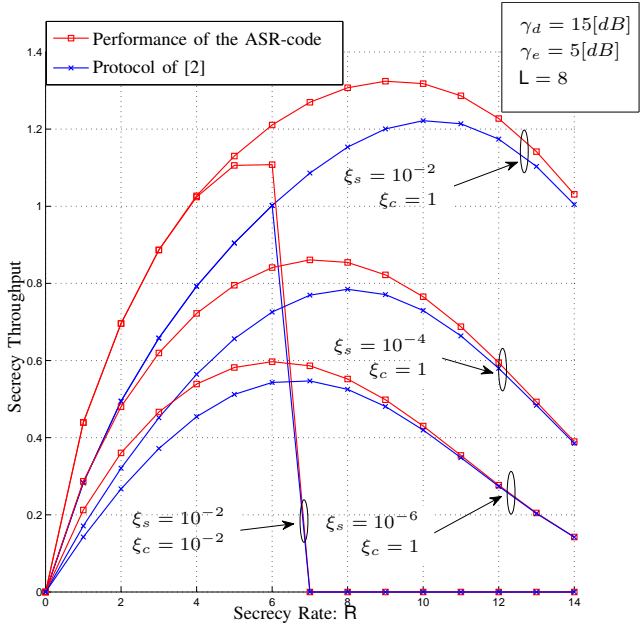}
\end{scriptsize}
\vspace{-0.9cm}
\caption{Secrecy throughput depending on the secrecy rate $\textsf{R}$, under different pairs of constraints $(\xi_{\sf{c}},\xi_{\sf{s}}) \in \big\{ (1, 10^{-2}) , (1, 10^{-4}) , (1, 10^{-6}) \big\}$. For each setting, the ASR-code outperforms the protocol of \cite{TangLiuSpasojevicPoor09}.}\label{fig:MultiPlotVaryingR0_2016_05_17}
    \end{center}
\end{figure}


\textbf{Trade-off between connection and secrecy outage probabilities.}\\ 
As mentioned in Sec \ref{sec:ChannelModel}, the constraints $\xi_{\sf{c}}$ and $\xi_{\sf{s}}$ are not always compatible. Fig. \ref{fig:PlotTradeoffOutagePsoPco_2016_04_15} represents the trade-off between the connection $\mc{P}_{\sf{co}}$ and the secrecy $\mc{P}_{\sf{so}}$ outages, depending on the maximal number of transmissions $L\in \{1,2,4,8\}$, for $\textsf{R} =0$. For each setting, the trade-off for the protocol of \cite{TangLiuSpasojevicPoor09} is more restrictive than for the ASR-code. Splitting the dummy-message rate over multiple transmissions, \ie with $\textsf{R}_2>0$, provides a small improvement for this trade-off. For a given pair of constraints $(\xi_{\sf{c}},\xi_{\sf{s}})$, there exists a minimal number of transmissions $L$ such that the connection and secrecy outage probabilities $\mc{P}_{\sf{co}}\leq \xi_{\sf{c}}$ and $\mc{P}_{\sf{so}}\leq \xi_{\sf{s}}$ satisfy the constraints.

\textbf{Range of dummy-message rate} $\textsf{R}_1\in [ \textsf{R}_1^{\min} ,  \textsf{R}_1^{\max} ]$.\\  
The minimal rate $\textsf{R}_1$ should guarantee that during the first transmission, the equation $\mc{P}\big(I(X_1;Z_1|k_1) \geq \textsf{R}_1\big)= \xi_{\sf{s}}$ is satisfied with equality. If the inequality was strict $\mc{P}\big(I(X_1;Z_1|k_1) \geq \textsf{R}_1\big)< \xi_{\sf{s}}$, then it would be possible to decrease the  rate parameter $\textsf{R}_1$ in order to increase the secrecy rate $\textsf{R}$ and the corresponding throughput. The minimal rate $\textsf{R}_1^{\min}\leq \textsf{R}_1 $ is defined by:
\begin{eqnarray}
 \textsf{R}_1^{\min} &=& \log_2\Big(1 - \gamma_{\textsf{e}} \cdot \log_2(\xi_{\sf{s}}) \Big) .
\end{eqnarray} 

The maximal  rate $\textsf{R}_1$ should guarantee that the secrecy outage probability for $L$ possible  transmissions, is equal to $\xi_{\sf{s}}$. A larger dummy-message rate $\textsf{R}_1$ would be a waste of transmission resources. 
This induces a maximal rate $\textsf{R}_1^{\max} \geq \textsf{R}_1$, defined by:
\begin{eqnarray}
 \textsf{R}_1^{\max} &\text{ s.t. }& \mc{P}\bigg(  \sum_{j\in 1}^L  I(X_j;Z_j|k_j) \geq \textsf{R}_1^{\max}  \bigg)= \xi_{\sf{s}} .
\end{eqnarray} 
The dummy-message rate $ \textsf{R}_1^{\max}$ is optimal for the protocol of \cite{TangLiuSpasojevicPoor09}, \ie where second rate $\textsf{R}_2=0$ is zero. 

  \begin{figure}[ht!]
  \begin{scriptsize}
\psfragscanon
\psfrag{aa}[l][B]{$\gamma_d = 15 [dB]$}
\psfrag{bb}[l][B]{$\gamma_e = 5 [dB]$}
\psfrag{cc}[l][B]{$\textsf{L}=8$}
\psfrag{r1}[r][]{$\textsf{R}_1^{\max}$ {optimal for \cite{TangLiuSpasojevicPoor09}}}
\psfrag{r2}[r][]{$\textsf{R}_1^{\star}$ optimal {for the ASR-code}}
\psfrag{r3}[r][]{$\textsf{R}_2^{\star}$ optimal {for the ASR-code}}
\psfrag{labelx}[][]{Secrecy Rate: $\textsf{R}$}
\psfrag{labely}[][]{Dummy-Message Rates}
\psfrag{datadatadatadata1}[l][l]{$\xi_s=10^{-6}$, $\xi_c=1$}
\psfrag{datadatadatadata2}[l][l]{$\xi_s=10^{-4}$, $\xi_c=1$}
\psfrag{datadatadatadata3}[l][l]{$\xi_s=10^{-2}$, $\xi_c=1$}
\includegraphics[width=0.5\textwidth]{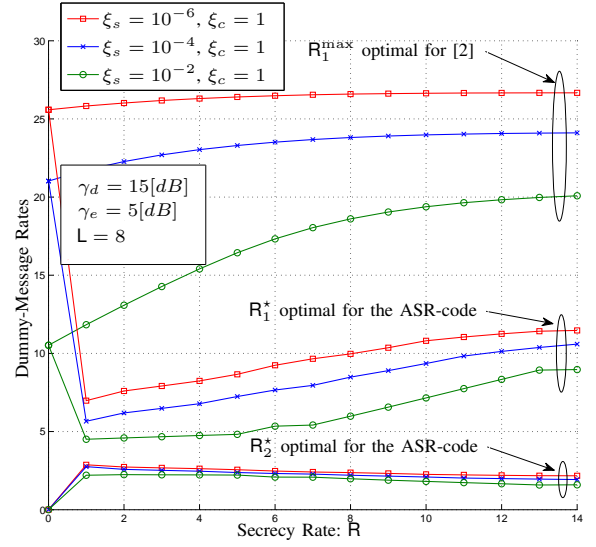}
\end{scriptsize}
\vspace{-0.9cm}
\caption{Optimal rates $(\textsf{R}_1^{\star},\textsf{R}_2^{\star})$ for the ASR-code and $\textsf{R}_1^{\max}$ for the protocol of \cite{TangLiuSpasojevicPoor09}, depending on the secrecy rate $\textsf{R}$ under different pairs of constraints $(\xi_{\sf{c}},\xi_{\sf{s}}) \in \big\{ (1, 10^{-2}) , (1, 10^{-4}) , (1, 10^{-6}) \big\}$.}\label{fig:MultiPlotVaryingR0B_2016_05_17}
\end{figure}

\textbf{Optimization of dummy-message rates} $(\textsf{R}_1,\textsf{R}_2)$.\\
We fixe the secrecy rate $\textsf{R}\geq0$ and for each rate $\textsf{R}_1^{\min} \leq \textsf{R}_1 \leq \textsf{R}_1^{\max}$, we find $\textsf{R}_2^{\star}(\textsf{R}_1)$ such that the secrecy outage constraint $\mc{P}_{\sf{so}} = \xi_{\sf{s}}$ is satisfied with equality.  Then, we optimize the secrecy throughput regarding the pair of rates $\big(\textsf{R}_1,\textsf{R}_2^{\star}(\textsf{R}_1)\big)$ and the secrecy rate $\textsf{R}$.

   \begin{figure}[ht!]
   \begin{scriptsize}
\psfragscanon
\psfrag{aa}[l][B]{$\gamma_d = 15 [dB]$}
\psfrag{bb}[l][B]{$\gamma_e = 5 [dB]$}
\psfrag{cc}[l][B]{$\textsf{L}=8$}
\psfrag{r1}[l][]{$\xi_s=10^{-2}$, $\xi_c=1$}
\psfrag{r2}[l][]{$\xi_s=10^{-4}$, $\xi_c=1$}
\psfrag{r3}[l][]{$\xi_s=10^{-6}$, $\xi_c=1$}
\psfrag{labelx}[][]{Secrecy Rate: $\textsf{R}$}
\psfrag{labely}[][]{Outage Probability}
\psfrag{datadatadatadatadatadatadatadatadatadatadata1}[l][l]{Connection Outage Probability {for the ASR-code}}
\psfrag{datadatadatadatadatadatadatadatadatadatadata2}[l][l]{Connection Outage Probability {for the protocol of \cite{TangLiuSpasojevicPoor09}}}
\psfrag{datadatadatadatadatadatadatadatadatadatadata3}[l][l]{Secrecy Outage Probability {for both ASR-code and \cite{TangLiuSpasojevicPoor09}}}
\includegraphics[width=0.5\textwidth]{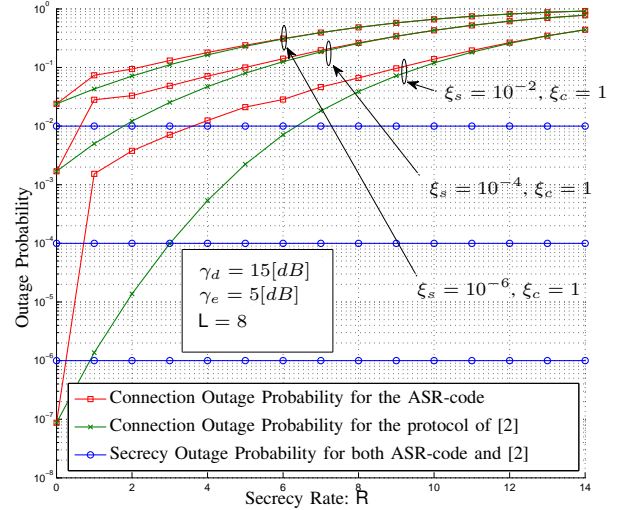}
\end{scriptsize}
\vspace{-0.9cm}
\caption{connection and secrecy outage probabilities for the ASR-code and for the protocol of \cite{TangLiuSpasojevicPoor09}, depending on the secrecy rate $\textsf{R}$ under different pairs of constraints $(\xi_{\sf{c}},\xi_{\sf{s}}) \in \big\{ (1, 10^{-2}) , (1, 10^{-4}) , (1, 10^{-6}) \big\}$.}\label{fig:MultiPlotVaryingR0C_2016_05_17}
\end{figure}


\begin{figure*}[ht!]
\begin{center}
\begin{tabular}{||r|c|c|c|c||}
\hline
\hline
Constraints $(\xi_{\sf{c}},\xi_{\sf{s}})$ &  $(1, 10^{-6})$&  $(1, 10^{-4})$ & $(1, 10^{-2})$ & $(10^{-2}, 10^{-2})$ \\
\hline\hline
Maximal secrecy throughput $\eta$ with $\textsf{R}_1^{\max}$,  $\textsf{R}_2=0$&  $0.55$ & $0.78$ &$1.22$ &  $1.00$\\
\hline
Maximal secrecy throughput  $\eta$ with $(\textsf{R}_1^{\star},\textsf{R}_2^{\star})$&  $0.60$  &  $0.86$ &    $1.32$  &     $1.11$\\  
\hline
Increase of secrecy throughput &  $9\%$ &   $10\%$ & $8\%$ & $11\%$ \\ 
\hline \hline
$\E\big[\textsf{L}\big]$  with $\textsf{R}_1^{\max}$, $\textsf{R}_2=0$ & $7.76$ & $7.57$ &   $7.20$ &   $5.94$ \\
\hline
$\E\big[\textsf{L}\big]$ with $(\textsf{R}_1^{\star},\textsf{R}_2^{\star})$ &  $6.92$ &    $6.53$ &    $6.14$ &    $5.36$\\
\hline
Reduction of exp. number of transmissions & $-10\%$ & $-14\%$ &  $-15\%$ &  $ -10\%$  \\
\hline\hline
\end{tabular}
\end{center}
\caption{Maximal secrecy throughput $\eta$ corresponding to Fig.  \ref{fig:MultiPlotVaryingR0_2016_05_17} and expected number of transmissions $\E\big[\textsf{L}\big]$. For each pair of outage parameters $(\xi_{\sf{c}},\xi_{\sf{s}})$, the ASR-code provides a higher secrecy throughput and a lower expected number of transmissions, compared to the protocol of \cite{TangLiuSpasojevicPoor09}. }\label{fig:tabular}
\end{figure*}

\textbf{Numerical Results}\\
Figure \ref{fig:MultiPlotVaryingR0_2016_05_17} compares the secrecy throughput for the ASR-code and for the protocol of \cite{TangLiuSpasojevicPoor09}. These two curves are drawn depending the secrecy rate $\textsf{R}\geq 0$, by considering four  pairs of constraints:\\ $(\xi_{\sf{c}},\xi_{\sf{s}}) \in \bigg\{ (1, 10^{-2}) , (1, 10^{-4}) , (1, 10^{-6}) , (10^{-2}, 10^{-2})\bigg\}$.
\begin{itemize}
\item[$\bullet$] As mentioned in Fig. \ref{fig:tabular}, splitting the dummy-message rate using $(\textsf{R}_1,\textsf{R}_2)$ improves the secrecy throughput by more than 8\%, compared to the approach of \cite{TangLiuSpasojevicPoor09}, with only one parameter $\textsf{R}_1^{\max}$, \ie with $\textsf{R}_2=0$. 
\item[$\bullet$] Tightening the secrecy constraint $\xi_{\sf{s}}$, reduces the secrecy throughput.
\item[$\bullet$] As mentioned for one transmission in Sec. \ref{sec:ChannelModel}, the connection outage constraint $\xi_{\sf{c}}$ induces a truncation of the secrecy throughput. This is illustrated by the curves corresponding to: $(\xi_{\sf{c}},\xi_{\sf{s}}) \in \big\{ (1, 10^{-2}) ,  (10^{-2}, 10^{-2})\big\}$.
\item[$\bullet$] The optimal rates $(\textsf{R}_1^{\star},\textsf{R}_2^{\star})$ for the ASR-code are presented in Fig.  \ref{fig:MultiPlotVaryingR0B_2016_05_17}.  As expected, the first parameter  $\textsf{R}_1^{\star} < \textsf{R}_1^{\max}$ is lower for the ASR-code than for the protocol of \cite{TangLiuSpasojevicPoor09}. Therefore, the first transmissions are more likely to be decoded correctly and this increases the secrecy throughput.
\item[$\bullet$] The connection outage probability $\mc{P}_{\sf{co}}$ corresponding to the optimal parameters $(\textsf{R},\textsf{R}_1^{\star},\textsf{R}_2^{\star})$ of the ASR-code are presented in Fig.  \ref{fig:MultiPlotVaryingR0C_2016_05_17}. For $(\xi_{\sf{c}},\xi_{\sf{s}})=(1, 10^{-2})$, the secrecy rate $\textsf{R}=6$ induces a connection outage probability close to $\mc{P}_{\sf{co}}\simeq 10^{-2}$ that corresponds to the truncation of the secrecy throughput for $\textsf{R}\geq6$, on Fig. \ref{fig:MultiPlotVaryingR0_2016_05_17}. The connection outage probability is larger for the ASR-code than for the protocol of \cite{TangLiuSpasojevicPoor09} because the total dummy-message rate $\textsf{R}_1 + (L-1)\cdot \textsf{R}_2 > \textsf{R}_1^{\max}$ is greater. However, this does not prevent the secrecy throughput of the ASR-code to be greater than for the protocol of \cite{TangLiuSpasojevicPoor09}.
\item[$\bullet$] The expected number of transmissions $\E\big[\textsf{L}\big]$ is represented in Fig.  \ref{fig:MultiPlotVaryingR0D_2016_05_17}. In Fig. \ref{fig:tabular}, the secrecy throughput and the expected number of transmissions $\E\big[\textsf{L}\big]$ are provided for the constraints: $(\xi_{\sf{c}},\xi_{\sf{s}}) \in \big\{ (1, 10^{-2}) , (1, 10^{-4}) , (1, 10^{-6}) , (10^{-2}, 10^{-2})\big\}$.
\item[$\bullet$] Compared to the protocol of  \cite{TangLiuSpasojevicPoor09}, the ASR-code increases the secrecy throughput $\eta$ by more than $8\%$ and reduces the expected number of transmissions $\E\big[\textsf{L}\big]$ by more than $10\%$.
\end{itemize}

  \begin{figure}[ht!]
  \begin{scriptsize}
\psfragscanon
\psfrag{aa}[l][B]{$\gamma_d = 15 [dB]$}
\psfrag{bb}[l][B]{$\gamma_e = 5 [dB]$}
\psfrag{cc}[l][B]{$\textsf{L}=8$}
\psfrag{labelx}[][]{Secrecy Rate: $\textsf{R}$}
\psfrag{labely}[][]{Expected number of transmission: $\E[\textsf{L}]$}
\psfrag{datadatadatadatadatadatadatadatadata1}[l][l]{$\xi_s = 10^{-6}$, $\xi_c = 1$ {for the protocol of \cite{TangLiuSpasojevicPoor09}}}
\psfrag{datadatadatadatadatadatadatadatadata2}[l][l]{$\xi_s = 10^{-6}$, $\xi_c = 1$ {for the ASR-code}}
\psfrag{datadatadatadatadatadatadatadatadata3}[l][l]{$\xi_s = 10^{-4}$, $\xi_c = 1$ {for the protocol of \cite{TangLiuSpasojevicPoor09}}}
\psfrag{datadatadatadatadatadatadatadatadata4}[l][l]{$\xi_s = 10^{-4}$, $\xi_c = 1$ {for the ASR-code}}
\psfrag{datadatadatadatadatadatadatadatadata5}[l][l]{$\xi_s = 10^{-2}$, $\xi_c = 1$ {for the protocol of \cite{TangLiuSpasojevicPoor09}}}
\psfrag{datadatadatadatadatadatadatadatadata6}[l][l]{$\xi_s = 10^{-2}$, $\xi_c = 1$ {for the ASR-code}}
\includegraphics[width=0.5\textwidth]{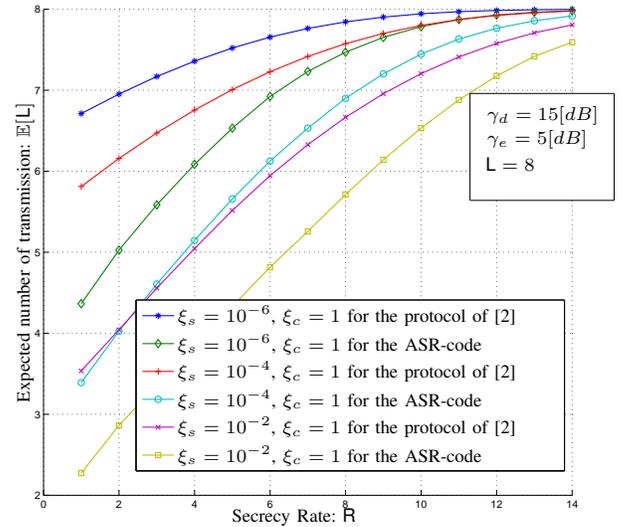}
\end{scriptsize}
\vspace{-0.9cm}
\caption{Expected number of transmissions $\E\big[\textsf{L}\big]$ for the ASR-code and for the protocol of \cite{TangLiuSpasojevicPoor09}, depending on the secrecy rate $\textsf{R}$ with different pairs of constraints $(\xi_{\sf{c}},\xi_{\sf{s}}) \in \big\{ (1, 10^{-2}) , (1, 10^{-4}) , (1, 10^{-6}) \big\}$.}\label{fig:MultiPlotVaryingR0D_2016_05_17}
\end{figure}

\section{Conclusion}\label{sec:conclusion}

We investigate secure HARQ protocols for state-dependent channels where the encoder only knows the statistics of the channel state. In this case, the reliability and security are defined in a probabilistic sense and there is a trade-off between the constraints we can impose on these two criteria. 

The presence of multiple transmissions rounds in HARQ offers new dimensions which we exploit in the design of the code to ensure secrecy and reliability. This was done in the literature, using a unique dummy-message. Our work follows this idea but, unlike previous works, we design a new code tailored for HARQ protocol, by splitting the dummy-message rate over several rate parameters.  These additional degrees of freedom improve the matching between the dummy-message rates and the realization of the eavesdropper channel. We evaluate the performance in terms of secrecy outage probability,  connection outage probability and secrecy throughput. For Rayleigh fading channel, the splitting of the dummy-message rate provides higher secrecy throughput and lower expected duration/average delay.

\appendices

\section{Proof of Theorem \ref{theo:RandomLeszekCode}}\label{sec:proofTheoLz}

We prove the Theorem  \ref{theo:RandomLeszekCode} considering $L=2$ transmissions. We provide a coding scheme that is reliable and secure for all pair of channel states $(k_1,k_2)$ that satisfy equation \eqref{eq:ChannelStates}.
\begin{eqnarray}
(k_1,k_2) \in \mc{S}^c_1(\varepsilon,\textsf{R},\textsf{R}_{1} ,\PP_{\sf{x}}^{\star}) \cap  \mc{S}_2(\varepsilon,\textsf{R},\textsf{R}_{1},\textsf{R}_{2},\PP_{\sf{x}}^{\star}).\label{eq:ChannelStates}
\end{eqnarray}
The first transmission is not reliable, the encoder receives a $\sf{NACK}_1$ feedback and starts a second transmission. More precisely, the channel states $(k_1,k_2)$ satisfy equations \eqref{eq:channelstates3Lz}, \eqref{eq:channelstates4Lz}, \eqref{eq:channelstates5Lz}, \eqref{eq:channelstates6Lz}.
\begin{eqnarray}
\textsf{R} + \textsf{R}_{1} + \textsf{R}_{2} &\leq&    I(X_1;Y_1|k_1)  +   I(X_2;Y_2|k_2) -  8 \varepsilon, \label{eq:channelstates3Lz} \\
\textsf{R} + \textsf{R}_{1} &>&    I(X_1;Y_1|k_1) -  4 \varepsilon, \label{eq:channelstates4Lz} \\
 \textsf{R}_{1} + \textsf{R}_{2}  &\geq&   I(X_1;Z_1|k_1)  +   I(X_2;Z_2|k_2)-  4  \varepsilon, \label{eq:channelstates5Lz} \\
 \textsf{R}_{1} &\geq&   I(X_1;Z_1|k_1)-  4 \varepsilon. \label{eq:channelstates6Lz}
 \end{eqnarray}
 Equations \eqref{eq:channelstates3Lz}, \eqref{eq:channelstates5Lz}, \eqref{eq:channelstates6Lz} correspond to the definition of the set of channel states $ \mc{S}_2(\varepsilon,\textsf{R},\textsf{R}_{1},\textsf{R}_{2},\PP_{\sf{x}}^{\star})$ and equation \eqref{eq:channelstates4Lz}  corresponds to the $\sf{NACK}_1$ feedback, \ie the first transmission failed $k_1 \notin \mc{S}^c_1(\varepsilon,\textsf{R},\textsf{R}_{1} ,\PP_{\sf{x}}^{\star}) $. Combining  \eqref{eq:channelstates3Lz} and \eqref{eq:channelstates4Lz}, it induces equation \eqref{eq:channelstates7Lz} that will be used in the following.
\begin{eqnarray}
 \textsf{R}_{2} &\leq&      I(X_2;Y_2|k_2) - 4 \varepsilon. \label{eq:channelstates7Lz}
\end{eqnarray}
Fig. \ref{fig:RateRegions} shows that equation \eqref{eq:channelstates7Lz} is a direct consequence of equation \eqref{eq:channelstates3Lz}, since the second transmission starts only when there is a $\textsf{Nack}_1$ feedback. Let  the length of the first transmission bloc $\bar{n}  \in \N$ be larger than $(n_1,n_2,n_3,n_4,n_5, n_6, n_7, n_8, n_9)$ given by equations (\ref{eq:ParametreN1Lz}), (\ref{eq:ParametreN2Lz}), (\ref{eq:ParametreN2Lzb}), (\ref{eq:ParametreN2Lzc}), (\ref{eq:Proof:ErrorE3Lz}), (\ref{eq:ParametreN4Lz}), (\ref{eq:ParametreN4Lzb}), (\ref{eq:ParametreN4Lzc}) and (\ref{eq:FifthTermebisLz}). We prove that there exists a HARQ-code ${c}^{\star} \in \mc{C}(n,\textsf{R} ,L)$ with stochastic encoder such that the error probability and the information leakage rate   satisfy equation \eqref{eq:Final4Lz}, for all channel states $(k_1,k_2) \in \mc{S}^c_1(\varepsilon,\textsf{R},\textsf{R}_{1} ,\PP_{\sf{x}}^{\star}) \cap  \mc{S}_2(\varepsilon,\textsf{R},\textsf{R}_{1},\textsf{R}_{2},\PP_{\sf{x}}^{\star})$,
\begin{eqnarray}
\PP_{\textsf{e}}\bigg({c}^{\star} \bigg|k_1,k_2 \bigg)  \leq  \varepsilon',\quad \mc{L}_{\textsf{e}}\bigg({c}^{\star} \bigg|k_1,k_2\bigg) \leq \varepsilon' , \label{eq:Final4Lz}
\end{eqnarray}
with $ \varepsilon' =  \varepsilon  \cdot ( 13 +   20   \log_2 |\mc{X} | )$.
 
Using similar arguments, the HARQ-code with stochastic encoder ${c}^{\star}  \in \mc{C}(n,\textsf{R},L)$  can be extended to the case of $L$ transmissions. The coding scheme is reliable and secure for all channel states $(k_1,\ldots , k_L) \in \bigcup_{l=1}^L \mc{S}_l(\varepsilon,\textsf{R},\textsf{R}_{1},\ldots,\textsf{R}_{L},\PP_{\sf{x}}^{\star})$ stated in definition \ref{def:ChannelStatesLz}.

\begin{figure}[!ht]
\begin{center}
\begin{footnotesize}
\psset{xunit=0.5cm,yunit=0.5cm}
\begin{pspicture}(-4,-1.7)(7,18)
\psellipse(2,16)(3,0.7)
\psellipse(2,14)(3,0.7)
\psdots(2,12.8)(2,12.5)(2,12.2)
\psellipse(2,11)(3,0.7)
\psdots(2,16)
\rput[u](4.5,17.1){$X_1^n(m,w_1)\sim  \PP_{\sf{x}}^{\star \times n } $}
\psline[linewidth=1pt]{-}(-2,16)(2,16)
\psline[linewidth=1pt]{-}(2,15)(2,16)
\rput[u](-2.4,16){$m$}
\rput[u](2.4,15){$w_1$}
\rput[u](-3.7,13.5){$|\mc{M} |= 2^{n \textsf{R}} $}
\psbrace[linecolor=black,ref=rC](-1,17)(-1,10){}
\rput[u](-2.5,18){$|\mc{M} \times \mc{M}_{1}| = 2^{n ( \textsf{R} + \textsf{R}_{1} )} $}
\rput[u](8,16){$| \mc{M}_{1} | = 2^{n \textsf{R}_{1} } $}
\put(3,8){\vector(-2,0){1}}
\psline[linewidth=2pt](-6,9)(10,9)
\psellipse(2,6)(3,0.7)
\pscircle(2,6){0.24}
\pscircle(1,6){0.24}
\pscircle(4,6){0.24}
\pscircle(0,6){0.24}
\psdots(2.7,6)(3,6)(3.3,6)
\psellipse(2,4)(3,0.7)
\psdots(2,2.8)(2,2.5)(2,2.2)
\psellipse(2,1)(3,0.7)
\pscircle(2,1){0.24}
\pscircle(1,1){0.24}
\pscircle(4,1){0.24}
\pscircle(0,1){0.24}
\psdots(2.7,1)(3,1)(3.3,1)
\psdots(2,6)
\rput[u](4.5,7.1){$X_2^{ n}(m,w_1,w_2) \sim  \PP_{\sf{x}}^{\star   { n } }$}
\psline[linewidth=1pt]{-}(-2,6)(2,6)
\psline[linewidth=1pt]{-}(2,5)(2,6)
\psline[linewidth=1pt]{-}(1,5)(2,6)
\rput[u](-2.4,6){$m$}
\rput[u](2.4,5){$w_1$}
\rput[u](0.6,5){$w_2$}
\rput[u](-3.7,3.5){$| \mc{M} | = 2^{ n\textsf{R}} $}
\psbrace[linecolor=black,ref=rC](-1,7)(-1,0){}
\rput[u](-2.5,8){$ | \mc{M} \times \mc{M}_{1} \times \mc{M}_{2} | = 2^{ n ( \textsf{R} + \textsf{R}_{1} + \textsf{R}_{2} )} $}
\rput[u](7.5,5){$ |\mc{M}_{1} \times \mc{M}_{2}| = 2^{ n ( \textsf{R}_{1} + \textsf{R}_{2} )} $}
\put(2,2.5){\vector(-1,1){0.35}}
\rput[u](2.5,-1.5){$ |\mc{M}_{2}| = 2^{ n  \textsf{R}_{2}} $}
\psbrace[linecolor=black,ref=rC](-0.5,0.3)(4.5,0.3){}
\rput[u](8,1){$ |\mc{M}_{1}| = 2^{ n  \textsf{R}_{1} } $}
\put(3,0.5){\vector(-1,0){1}}
\end{pspicture}
\end{footnotesize}
\caption{Multilevel random coding scheme $C  \in\mc{C}(n,\textsf{R},L)$ stated in section \ref{sec:RandomCodeLz} for $L=2$ transmissions. The parameters $n\in \N$, $\in \R^+$, $\textsf{R}\in \R^+$, $\textsf{R}_{1} \in \R^+$, $\textsf{R}_{2} \in \R^+$ determine the cardinalities of the set of messages $ |\mc{M}| = 2^{n \textsf{R}} $, the cardinality of the bins 
$ |\mc{M}_{1} | = 2^{n \textsf{R}_{1}}$ and the number of sub-bins $|\mc{M}_{2} | = 2^{n \textsf{R}_{2}} $. The random codewords $X_1^n(m,w_1)$ and $X_2^{ n}(m,w_1,w_2)   $ are generated with i.i.d. probability distribution $\PP_{\sf{x}}^{\star \times n } $. }\label{figure:BinningLeszekCode}
\end{center}
\end{figure}



 \subsection{Random HARQ-Code}\label{sec:RandomCodeLz}
 
We consider a random HARQ-code $C \in \mc{C}(n,\textsf{R},L)$ with stochastic encoder, represented by figure \ref{figure:BinningLeszekCode} for $L=2$ transmissions and defined as follows:
\begin{itemize}
\item[$\bullet$] \textit{Random codebook for the first transmission.} Generate $|\mc{M} \times \mc{M}_{1}|= 2^{n ( \textsf{R} + \textsf{R}_{1} )} $ sequences $X_1^n \in \mc{X} $ drawn from the probability distribution $\PP_{\sf{x}}^{\star \times n}$. Randomly bin them into $|\mc{M} |= 2^{n \textsf{R}} $ bins denoted by  $m\in \mc{M}$, each of them containing  $|\mc{M}_{1} |= 2^{n \textsf{R}_{1} } $  sequences $X_1^n \in \mc{X}^n$ indexed by the parameter $w_1 \in \mc{M}_{1}$. 
\item[$\bullet$] \textit{Encoding for the first transmission.} The encoder observes the realization of the message $m \in \mc{M}$. It chooses at random the parameter $w_1 \in \mc{M}_{1}$ using the uniform probability distribution and sends the sequence of channel inputs $x_1^{n}(m,w_1)$ through the channel $\mc{T}_1$. 
\item[$\bullet$] \textit{First feedback from the decoder.} The decoder observes the realization of the channel state $k_1 \in \mc{K}_1$ and sends to the encoder the feedback "$\textsf{Ack}_1$" if it can decode the message after the first transmission (i.e. equation (\ref{eq:channelstates4Lz}) is not satisfied).  It sends the feedback "$\textsf{Nack}_1$" if it can not decode after the first transmission (i.e. equation (\ref{eq:channelstates4Lz}) is satisfied).
\item[$\bullet$] \textit{Decoding fonction for "$\textsf{Ack}_1$".} The decoder observes the state parameter $k_1\in \mc{K}_1 $ and finds the pair of indices $(m,w_1) \in \mc{M} \times \mc{M}_{1} $ such that $x_1^n(m,w_1) \in A_{\varepsilon}^{{\star}{n}}(y_1^n|k_1)$ is jointly typical with the sequence of channel outputs. Its returns the index $m\in \mc{M}$ of the transmitted message.
\item[$\bullet$] \textit{Random codebook for the second transmission.} Generate $|\mc{M} \times \mc{M}_{1} \times \mc{M}_{2} |= 2^{ n ( \textsf{R} + \textsf{R}_{1} + \textsf{R}_{2} )} $ sequences $X_2^{ n} \in \mc{X}^{ n}  $ drawn from the  probability distribution $\PP_{\sf{x}}^{\star \times  n}$. Randomly bin them into $|  \mc{M}|= 2^{ n  \textsf{R}} $ bins denoted by  $m\in \mc{M}$, each of them containing  $| \mc{M}_{1} \times \mc{M}_{2} |= 2^{ n ( \textsf{R}_{1} + \textsf{R}_{2} )}  $  sequences $X_2^{ n} \in \mc{X}^{ n}  $ indexed by a pair of parameters $(w_1,w_2) \in \mc{M}_{1} \times \mc{M}_{2}$. Each bin $m\in \mc{M}$ is divided into $|\mc{M}_{2} |= 2^{ n  \textsf{R}_{2}} $ sub-bins containing $| \mc{M}_{1} |= 2^{ n  \textsf{R}_{1} } $ sequences $X_2^{ n} \in \mc{X}^{ n}  $. We denote by $w_2\in \mc{M}_{2}$ the index of the sub-bins and by $w_1\in \mc{M}_{1}$ the index of the sequence of symboles $X_2^{ n}(m,w_1,w_2) \in \mc{X}^{ n}  $.
\item[$\bullet$] \textit{Encoding for the second transmission.} If the encoder receives a "$\textsf{Nack}_1$"  feedback, the second transmission starts.
Encoder chooses at random the parameter $w_2 \in \mc{M}_{2}$ using the uniform probability distribution and sends  the sequence of channel inputs $x_2^{ n}(m,w_1,w_2)$.
\item[$\bullet$] \textit{Second feedback from the decoder.} The decoder observes the realization of the channel state $(k_1,k_2)\in \mc{K}_1 \times \mc{K}_2$ and sends to the encoder the feedback "$\textsf{Ack}_2$" if it can decode   (i.e. equation (\ref{eq:channelstates3Lz}) is satisfied). It sends the feedback "$\textsf{Nack}_2$" if it can not decode (i.e. equation (\ref{eq:channelstates3Lz}) is satisfied).
\item[$\bullet$] \textit{Decoding function for "$\textsf{Ack}_2$".} The decoder observes the state parameters $(k_1,k_2)\in \mc{K}_1 \times \mc{K}_2$ and finds the triple of indices $(m,w_1,w_2) \in \mc{M} \times \mc{M}_{1} \times \mc{M}_{2}$ such that $x_1^n(m,w_1) \in A_{\varepsilon}^{{\star}{n}}(y_1^n|k_1)$ is jointly typical with the sequence of outputs of the first channel $\mc{T}_1$ and such that $x_2^{ n}(m,w_1,w_2) \in A_{\varepsilon}^{{\star}{ n}}(y_2^{ n}|k_2)$ is jointly typical with the sequence of outputs of the second channel $\mc{T}_2$. Its returns the index $m\in \mc{M}$ of the transmitted message.
\item[$\bullet$] \textit{Larger number of transmissions $L>2$.} The same procedure involving random codebook, encoding, feedbacks and decoding  is repeated for $L>2$ transmissions.
\item[$\bullet$] \textit{An error} is declared when the sequences $(x_1^n, y_1^n,z_1^n)\notin A_{\varepsilon}^{{\star}{n}}(\QQ_1|k_1)$ or $(x_2^{ n},y_2^{ n}, z_2^{ n})\notin A_{\varepsilon}^{{\star}{{ n}}}(\QQ_2|k_2)$ are not jointly typical for the probability distributions $\QQ_1 = \PP_{\sf{x}}^{\star} \times \mc{T}_1 \in \Delta(\mc{X} \times \mc{Y}_1 \times \mc{Z}_1 )$ and $\QQ_2 = \PP_{\sf{x}}^{\star} \times \mc{T}_2 \in \Delta(\mc{X} \times \mc{Y}_2 \times \mc{Z}_2 )$.
\end{itemize}


\begin{remark}
The parameter $n\in \N$ is the length of the transmission block, $ |\mc{M}|= 2^{n \textsf{R}}$ is the cardinality of the set of messages $\mc{M}$, $|\mc{M}_{1}|= 2^{n \textsf{R}_{1}} $ is the cardinality of the set of dummy-messages $\mc{M}_{1}$ for the first transmission and $ |\mc{M}_{2}|= 2^{n \textsf{R}_{2}} $ is the cardinality of the set of dummy-messages $\mc{M}_{2}$ for the second transmission. We denote by $\PP_{\sf{x}}^{\star} \in \Delta(\mc{X})$  the probability distribution of the sequences of channel inputs.
\end{remark}

\subsection{Expected error probability} \label{sec:Proof:EspErrorProb}

We upper bound the expected error probability for fixed messages $(m,w_1,w_2)$ and channel states $(k_1,k_2) \in  \mc{S}^c_1(\varepsilon,\textsf{R},\textsf{R}_{1} ,\PP_{\sf{x}}^{\star}) \cap  \mc{S}_2(\varepsilon,\textsf{R},\textsf{R}_{1},\textsf{R}_{2},\PP_{\sf{x}}^{\star})$.
\begin{eqnarray}
&&\E_c\bigg[\PP\bigg( \bigg\{(x_1^n, y_1^n,z_1^n)\notin A_{\varepsilon}^{{\star}{n}}(\QQ_1|k_1) \bigg\} \nonumber\\
&&\qquad  \cup\bigg\{ (x_2^{ n},y_2^{ n}, z_2^{ n})\notin A_{\varepsilon}^{{\star}{{ n}}}(\QQ_2|k_2)\bigg\}
\bigg)\bigg] \leq   \varepsilon,\label{eq:ParametreN1Lz}\\
&&\E_c\bigg[\PP\bigg(\bigg\{\exists (m', w_1', w_2') \neq (m, w_1,w_2),\text{ s.t. } \nonumber\\
&& \qquad  \{ x_1^n(m', w_1')\in   A_{\varepsilon}^{{\star}{n}}(y_1^n|k_1) \}\nonumber\\
&&  \qquad \cap\{ x_2^{ n}(m', w_1',w_2')\in  A_{\varepsilon}^{{\star}{{ n}}}(y_2^{ n}|k_2)\}\bigg\}\bigg)\bigg] \leq   \varepsilon,\qquad \label{eq:ParametreN2Lz}\\
&&\E_c\bigg[\PP\bigg(\bigg\{\exists (m', w_1') \neq (m,w_1),\text{ s.t. }\nonumber\\
&& \qquad\{ x_1^n(m', w_1')\in A_{\varepsilon}^{{\star}{n}}(y_1^n|k_1) \}\nonumber\\
&& \qquad\cap\{ x_2^{ n}(m', w_1',w_2)\in  A_{\varepsilon}^{{\star}{{ n}}}(y_2^{ n}|k_2)\}\bigg\}\bigg)\bigg] \leq \varepsilon, \label{eq:ParametreN2Lzb}\\
&&\E_c\bigg[\PP\bigg(\bigg\{\exists  w_2' \neq w_2,\text{ s.t. } \nonumber\\
&& \qquad\qquad  x_2^{ n}(m, w_1,w_2')\in  A_{\varepsilon}^{{\star}{{ n}}}(y_2^{ n}|k_2)\bigg\}\bigg)\bigg] \leq \varepsilon. \label{eq:ParametreN2Lzc}
\end{eqnarray}
\eqref{eq:ParametreN1Lz} comes from the typical sequences \cite[pp. 26]{ElGammalKim(book)11}.\\
\eqref{eq:ParametreN2Lz} comes from \eqref{eq:channelstates3Lz} and \cite[pp. 46, Packing Lemma]{ElGammalKim(book)11} since the codewords $(X _1^n(m', w_1'),X _2^{ n}(m',  w_1', w_2'))$ are independent of $(X _1^n(m,w_1), X _2^{ n}(m, w_1,w_2))$.\\
\eqref{eq:ParametreN2Lzb} comes from \eqref{eq:channelstates3Lz} and \cite[pp. 46, Packing Lemma]{ElGammalKim(book)11}.\\
\eqref{eq:ParametreN2Lzc} comes from \eqref{eq:channelstates7Lz} and \cite[pp. 46, Packing Lemma]{ElGammalKim(book)11}.

This provides an upper bound on:
\begin{eqnarray}
 \E_c\bigg[\PP_{\textsf{e}}\bigg(C\bigg|k_1,k_2 \bigg) \bigg] &\leq& 4 \varepsilon.
\end{eqnarray}

\subsection{Expected information leakage rate}\label{sec:ProofEspTauxEquivoqueLz}
We provide an upper bound on the expected information leakage rate that is valid for all channel states $(k_1,k_2) \in  \mc{S}^c_1(\varepsilon,\textsf{R},\textsf{R}_{1} ,\PP_{\sf{x}}^{\star}) \cap  \mc{S}_2(\varepsilon,\textsf{R},\textsf{R}_{1},\textsf{R}_{2},\PP_{\sf{x}}^{\star})$. To this purpose, we introduce four auxiliary random variables  $V_1$, $J_1$, $V_2$ and $J_2$ that belong to the sets $\mc{M}_{\sf{V}_1}$, $\mc{M}_{\sf{J}_1}$, $\mc{M}_{\sf{V}_2}$ and $\mc{M}_{\sf{J}_2}$
 with cardinality $ |\mc{M}_{\sf{V}_1}| = 2^{n \textsf{R}_{\sf{V}_1}}$, $ |\mc{M}_{\sf{J}_1}| = 2^{n \textsf{R}_{\sf{J}_1}}$, $ |\mc{M}_{\sf{V}_2}| = 2^{n \textsf{R}_{\sf{V}_2}}$ and $|\mc{M}_{\sf{J}_2}| = 2^{n \textsf{R}_{\sf{J}_2}}$ given by:
 \begin{eqnarray}
\textsf{R}_{\sf{V}_1} &=&  I(X_1; Z_1 |k_1)  +  I(X_2 ; Z_2 |k_2)  \nonumber\\
&-& \min\bigg( I(X_2 ; Z_2 |k_2) , \textsf{R}_{2} \bigg)-  4 \varepsilon,\label{eq:ExpectedInfoLeakage1}\\
\textsf{R}_{\sf{V}_2} &=&   \min\bigg( I(X_2 ; Z_2 |k_2) , \textsf{R}_{2} \bigg) -  4 \varepsilon,\label{eq:ExpectedInfoLeakage3}\\
\textsf{R}_{\sf{J}_1} &=&  \textsf{R}_{1}  - \textsf{R}_{\sf{V}_1}\nonumber \\
&=& \min\bigg(\textsf{R}_{1}  - I(X_1; Z_1 |k_1) +  4 \varepsilon , \textsf{R}_{1} + \textsf{R}_{2}  \nonumber\\
&-& I(X_1; Z_1 |k_1) -  I(X_2 ; Z_2 |k_2) +  4 \varepsilon\bigg) ,\label{eq:ExpectedInfoLeakage2}\\
\textsf{R}_{\sf{J}_2} &=&  \textsf{R}_{2} -  \textsf{R}_{\sf{V}_2} \nonumber \\
&=& \max\bigg(\textsf{R}_{2} -   I(X_2 ; Z_2 |k_2) , 0 \bigg) +  4 \varepsilon .\label{eq:ExpectedInfoLeakage4}
\end{eqnarray}
The idea of this proof is to adapt the size of the set of dummy-messages to the realizations of the mutual informations $I(X_1; Z_1 |k_1)$ and $I(X_2 ; Z_2 |k_2)$. The parameters $\textsf{R}_{\sf{V}_1}$, $\textsf{R}_{\sf{V}_2}$ and $\textsf{R}_{\sf{J}_2}$ are positive. Equations (\ref{eq:channelstates5Lz}) and (\ref{eq:channelstates6Lz}) guarantees that parameter $\textsf{R}_{\sf{J}_1} $ is positive for all channel states $(k_1,k_2) \in  \mc{S}^c_1(\varepsilon,\textsf{R},\textsf{R}_{1} ,\PP_{\sf{x}}^{\star}) \cap  \mc{S}_2(\varepsilon,\textsf{R},\textsf{R}_{1},\textsf{R}_{2},\PP_{\sf{x}}^{\star})$. 
In this section, each bin ${m}\in \mc{M}$ is re-organized as follows:
\begin{itemize}
\item[$\bullet$] First, we divide each sub-bin $w_2\in \mc{M}_{2}$ of size $2^{n\textsf{R}_{1} }$ into $2^{n\textsf{R}_{\sf{J}_1} }$ sub-sub-bins of size $2^{n\textsf{R}_{\sf{V}_1}}$. 
\item[$\bullet$] Second, we concatenate  the sub-bins $w_2\in \mc{M}_{2}$ into $2^{n\textsf{R}_{\sf{J}_2} }$ super-sub-bins containing $2^{n\textsf{R}_{\sf{V}_2} }$ sub-bins $w_2\in \mc{M}_{2}$. 
\end{itemize}
This analysis does not modify the random code $C $ but it allows to provide an upper bound over the information leakage rate. The parameters $W_1$ and $W_2$ correspond to the pairs of auxiliary random variables  $W_1= (V_1,J_1)$ and $W_2= (V_2,J_2)$. 
\begin{eqnarray}
&&n \cdot \E_c\bigg[\mc{L}_{\textsf{e}}\bigg(C \bigg|k_1,k_2\bigg) \bigg]\nonumber \\
&=& I(M , W_1 , W_2 ; Z_1^n, Z_2^{ n} |C,k_1,k_2)  \label{eq:term1}\\
&-&  H( W_1 , W_2 | M, C,k_1,k_2)  \label{eq:term2}\\
&+& H( W_1 , W_2 | M, C, Z_1^n, Z_2^{ n},k_1,k_2) \label{eq:term3}.
\end{eqnarray}

$\bullet$ The first term \eqref{eq:term1} satisfies:
\begin{eqnarray}
&& I(M , W_1 , W_2, C ; Z_1^n, Z_2^{ n} |k_1,k_2)  \nonumber\\
&\leq& I(X_1^n, X_2^{ n} ; Z_1^n, Z_2^{ n} |k_1,k_2)   \label{eq:Equivoque4Lz}\\
&=& n\cdot( I(X_1; Z_1 |k_1) + \cdot I(X_2 ; Z_2 |k_2))  .\label{eq:Equivoque6Lz}
\end{eqnarray}
\eqref{eq:Equivoque4Lz} comes from the Markov chain $( C, M , W_1, W_2) -\!\!\!\!\minuso\!\!\!\!- (X_1^n, X_2^{ n} ) -\!\!\!\!\minuso\!\!\!\!- (Z_1^n, Z_2^{ n} )$ for all channel states $(k_1,k_2) \in  \mc{S}^c_1(\varepsilon,\textsf{R},\textsf{R}_{1} ,\PP_{\sf{x}}^{\star}) \cap  \mc{S}_2(\varepsilon,\textsf{R},\textsf{R}_{1},\textsf{R}_{2},\PP_{\sf{x}}^{\star})$.\\
\eqref{eq:Equivoque6Lz} comes from the independent generation of the sequences $X_1^n$ and $X_2^{ n}$ with i.i.d. probability distributions $\PP_{\sf{x}}^{\star}$.

$\bullet$ The second term \eqref{eq:term2} satisfies:
\begin{eqnarray}
H( W_1 , W_2 | M, C,k_1,k_2) = n\cdot( \textsf{R}_{1}  +   \textsf{R}_{2}). \label{eq:Equivoque7Lz}
\end{eqnarray}
\eqref{eq:Equivoque7Lz} comes from the fact that the random variable $ W_1$ and $W_2$ are drawn independently of  $(M, C,k_1,k_2)$ and uniformly distributed over the sets $\mc{M}_{1}$, $\mc{M}_{2}$ of cardinality $2^{n\textsf{R}_{1}}$,  $2^{n\textsf{R}_{2}}$.

$\bullet$ The third term \eqref{eq:term3} satisfies:
\begin{eqnarray}
&&H( W_1 , W_2 | M, C, Z_1^n, Z_2^{ n},k_1,k_2)\nonumber\\
&=& H(V_1,J_1, V_2,J_2| M, C, Z_1^n, Z_2^{ n},k_1,k_2) \label{eq:Equivoque8Lz}\\
&=&  H(J_1,J_2| M, C, Z_1^n, Z_2^{ n},k_1,k_2) \nonumber\\
&+& H(V_1, V_2 | J_1,J_2,M, C, Z_1^n, Z_2^{ n},k_1,k_2) 
\label{eq:Equivoque9Lz}\\
&\leq&  n\cdot (\textsf{R}_{\textsf{J}_1} + \textsf{R}_{\textsf{J}_2})   \nonumber\\
&+& H(V_1, V_2 | J_1,J_2,M, C, Z_1^n, Z_2^{ n},k_1,k_2) 
\label{eq:Equivoque10Lz}\\
&=&    n\cdot\bigg( \textsf{R}_{1}  +   \textsf{R}_{2}  - I(X_1; Z_1 |k_1)  -  I(X_2 ; Z_2 |k_2 )   +  8 \varepsilon \bigg)  \nonumber \\  
&+& H(V_1, V_2 | J_1,J_2,M, C, Z_1^n, Z_2^{ n},k_1,k_2) 
\label{eq:Equivoque11Lz}\\
&\leq&  n\cdot\bigg( \textsf{R}_{1}  +   \textsf{R}_{2} -  I(X_1; Z_1 |k_1)  -  I(X_2 ; Z_2 |k_2 )  \nonumber \\
&+&  \varepsilon \cdot \big(9 + 20 \log_2 |\mc{X} |   \big)\bigg). \label{eq:Equivoque13Lz}
\end{eqnarray}
\eqref{eq:Equivoque8Lz} comes from replacing indices $(w_1,w_2) \in \mc{M}_{1} \times \mc{M}_{2} $ by auxiliary indices $(v_1,j_1,v_2,j_2) \in \mc{M}_{\sf{V}_1} \times \mc{M}_{\sf{J}_1}  \times \mc{M}_{\sf{V}_2} \times \mc{M}_{\sf{J}_2}$.\\ 
\eqref{eq:Equivoque9Lz} and \eqref{eq:Equivoque10Lz} come from the properties of the entropy function and the cardinalities $|\mc{M}_{\sf{J}_1}| = 2^{n \textsf{R}_{\sf{J}_1}}$ and $|\mc{M}_{\sf{J}_2}| = 2^{n \textsf{R}_{\sf{J}_2}}$.\\
\eqref{eq:Equivoque11Lz} comes from the equations \eqref{eq:ExpectedInfoLeakage2} and \eqref{eq:ExpectedInfoLeakage4},
satisfied for all channel states $(k_1,k_2) \in  \mc{S}^c_1(\varepsilon,\textsf{R},\textsf{R}_{1} ,\PP_{\sf{x}}^{\star}) \cap  \mc{S}_2(\varepsilon,\textsf{R},\textsf{R}_{1},\textsf{R}_{2},\PP_{\sf{x}}^{\star})$ and the  equation: $\max(\sf{a},\sf{b}) + \min(\sf{a},\sf{b}) = \sf{a} + \sf{b}$. \\
 \eqref{eq:Equivoque13Lz} comes from Lemma \ref{lemma:FanoEquivocationLz}, that is based on Fano's inequality.

Equations \eqref{eq:Equivoque6Lz}, \eqref{eq:Equivoque7Lz} and \eqref{eq:Equivoque13Lz} provide an upper bound on:
\begin{eqnarray}
\E_c\Big[\mc{L}_{\textsf{e}}\Big(C \Big|k_1,k_2\Big) \Big] 
&\leq&  \varepsilon \cdot \big(9 + 20 \log_2 |\mc{X} |   \big).
\end{eqnarray}
This analysis can be extended to the case of $L>2$ transmissions by introducing the random variables $\textsf{R}_{\sf{V}_L}$ and $\textsf{R}_{\sf{J}_L}$.



\begin{lemma}\label{lemma:FanoEquivocationLz}
Fano's inequality provides the upper bound:
\begin{eqnarray}
&&H(V_1, V_2 | J_1,J_2,M, C, Z_1^n, Z_2^{ n},k_1,k_2) \nonumber\\
&\leq& n\cdot \bigg( \varepsilon  + 20\varepsilon   \cdot \log_2 |\mc{X} | \bigg).\label{eq:EqFanoLz}
\end{eqnarray}
\end{lemma}

\begin{proof}[Lemma \ref{lemma:FanoEquivocationLz}]
Suppose that the eavesdropper implements the decoding $g_{\sf{e}}$ defined by equation (\ref{eq:FifthTermeDecodageLz}) as follows:
\begin{itemize}
\item[$\bullet$] \textit{Decoding of the eavesdropper} $g_{\sf{e}}$ takes the sequence of channel outputs $Z_1^n \in \mc{Z}_1^n$, $ Z_2^{ n} \in \mc{Z}_2^{ n}$, the message $M \in \mc{M}$, the indices $J_1 \in  \mc{M}_{\sf{J}_1}$, $J_2 \in  \mc{M}_{\sf{J}_2}$ and the HARQ-code $C \in \mc{C}(n,\textsf{R},L)$ and returns the indices $V_1 \in \mc{M}_{\textsf{V}_1}$, $V_2 \in \mc{M}_{\textsf{V}_2}$ and the sequences $X_1^n(M,V_1,J_1)   \in \mc{X}^n$ and $X_2^{ n}(M,V_1,J_1,V_2,J_2)  \in \mc{X}^{ n}$ that are jointly typical with $Z_1^n \in \mc{Z}_1^n$ and $ Z_2^{ n} \in \mc{Z}_2^{ n}$.
\end{itemize}
\begin{eqnarray}
&g_{\sf{e}} :&  \mc{Z}_1^n \times \mc{Z}_2^{ n} \times \mc{M} \times \mc{M}_{\textsf{J}_1} \times \mc{M}_{\textsf{J}_2}   \nonumber  \\
&&  \times  \mc{K}_1 \times \mc{K}_2\times  \mc{C}(n,\textsf{R},\textsf{R}_\textsf{W},\textsf{R}_\textsf{L}, \PP_{\textsf{x}_1}^{\star} ,  \PP_{\textsf{x}_2}^{\star} )  
\nonumber  \\
&&\qquad \longrightarrow  \mc{X}^n  \times  \mc{X}^{ n} \times \mc{M}_{\textsf{V}_1} \times \mc{M}_{\textsf{V}_2}  .\label{eq:FifthTermeDecodageLz}
\end{eqnarray}
An error occurs if this decoding function $g_{\sf{e}}$ returns sequences of inputs and indices $(\hat{x}_1^n,\hat{x}_2^{ n}, \hat{v}_1,  \hat{v}_2)\neq g_{\sf{e}}({z}_1^n, z_2^{ n},m,j_1,j_2,c,k_1,k_2)$ that are different from the original tuple $({x}_1^n,{x}_2^{ n}, v_1,v_2)$. We provide an upper bound over the expected error probability of this decoding function $g_{\sf{e}}$.
\begin{eqnarray}
&&\E_c\bigg[\PP\bigg(  \bigg\{ (x_1^n, z_1^n)\notin A_{\varepsilon}^{{\star}{n}}(\QQ_1|k_1) \bigg\} \nonumber\\
&&\qquad \cup\bigg\{ (x_2^{ n}, z_2^{ n})\notin A_{\varepsilon}^{{\star}{{ n}}}(\QQ_2|k_2) \bigg\}
\bigg)\bigg] \leq   \varepsilon,\label{eq:Proof:ErrorE3Lz}\\
&&\E_c\bigg[\PP\bigg(\bigg\{\exists (v'_1, v'_2)  \neq (v_1, v_2) , \text{ s.t. } \nonumber\\
&&\qquad \{ x_1^n(m, v'_1, j_1)\in  A_{\varepsilon}^{{\star}{n}}(z_1^n|k_1) \} \nonumber\\
&&\quad\cap\{x_2^{ n}(m, v'_1, j_1,v'_2,j_2)\in  A_{\varepsilon}^{{\star}{n}}(z_2^{ n}|k_2) \} \bigg\}\bigg)\bigg] \leq   \varepsilon,\qquad \label{eq:ParametreN4Lz}\\
&&\E_c\bigg[\PP\bigg(\bigg\{\exists v'_1  \neq v_1, \text{ s.t. }  \nonumber\\
&&\qquad \{ x_1^n(m, v'_1,j_1)\in  A_{\varepsilon}^{{\star}{n}}(z_1^n|k_1) \} \nonumber\\
&&\quad\cap\{x_2^{ n}(m, v'_1, j_1,v_2,j_2)\in  A_{\varepsilon}^{{\star}{n}}(z_2^{ n}|k_2) \} \bigg\}\bigg)\bigg] \leq \varepsilon, \label{eq:ParametreN4Lzb}\\
&&\E_c\bigg[\PP\bigg(\bigg\{\exists v'_2  \neq v_2 , \text{ s.t. } \nonumber\\
&&\qquad x_2^{ n}(m, v_1, j_1,v'_2,j_2)\in  A_{\varepsilon}^{{\star}{n}}(z_2^{ n}|k_2)  \bigg\}\bigg)\bigg] \leq \varepsilon. \label{eq:ParametreN4Lzc}
\end{eqnarray}
\eqref{eq:Proof:ErrorE3Lz} comes from properties of typical sequences  \cite[pp. 26]{ElGammalKim(book)11}.\\
\eqref{eq:ParametreN4Lz} comes from \eqref{eq:ExpectedInfoLeakage1}, \eqref{eq:ExpectedInfoLeakage3}  and \cite[pp. 46, Packing Lemma]{ElGammalKim(book)11}.\\
\eqref{eq:ParametreN4Lzb} comes from \eqref{eq:ExpectedInfoLeakage1} and \cite[pp. 46, Packing Lemma]{ElGammalKim(book)11}.\\
\eqref{eq:ParametreN4Lzc} comes from \eqref{eq:ExpectedInfoLeakage3} and \cite[pp. 46, Packing Lemma]{ElGammalKim(book)11}.

Equations \eqref{eq:Proof:ErrorE3Lz}, \eqref{eq:ParametreN4Lz}, \eqref{eq:ParametreN4Lzb} and \eqref{eq:ParametreN4Lzc} prove that the expected probability of this decoding $g_e$ is upper bounded by $4 \varepsilon$. 
\begin{eqnarray}
&&H(V_1,  V_2 | M, J_1,  J_2,  C, Z_1^n, Z_2^{ n},k_1,k_2) \nonumber\\
&\leq& n\cdot \bigg( \varepsilon  + 20\cdot  \varepsilon \cdot   \log_2 |\mc{X} |   \bigg).\label{eq:FifthTermebisLz}
\end{eqnarray}
Equation \eqref{eq:FifthTermebisLz} comes from \cite[pp. 19, Fano's Inequality]{ElGammalKim(book)11} and $n\geq n_9 = \frac{1}{\varepsilon}$ and equations \eqref{eq:ExpectedInfoLeakage1} and \eqref{eq:ExpectedInfoLeakage3} which imply that $\log_2|\mc{M}_{\sf{V}_1}  |\leq  2 n\cdot \log_2|\mc{X} | $ and $\log_2|\mc{M}_{\sf{V}_2}  | \leq   n\cdot \log_2|\mc{X} |$.
\end{proof}

\vspace{-0.35cm}

\subsection{Conclusion}\label{sec:ConclusionLz}

For all $\varepsilon>0$, there exists $\bar{n}$, for all $n\geq \bar{n}$, there exists HARQ-code $c^{\star}\in \mc{C}(n,\textsf{R},L)$ such that $\PP_{\textsf{e}}\big(c^{\star} \big|k_1,k_2\big)   \leq  \varepsilon$ and $\mc{L}_{\textsf{e}}  \big(c^{\star} \big|k_1,k_2\big)     \leq  \varepsilon$, for all $(k_1,k_2) \in \mc{S}^c_1(\varepsilon,\textsf{R},\textsf{R}_{1} ,\PP_{\sf{x}}^{\star}) \cap  \mc{S}_2(\varepsilon,\textsf{R},\textsf{R}_{1},\textsf{R}_{2},\PP_{\sf{x}}^{\star})$.

\section{Proof of Proposition \ref{prop:SecrecyOutage}}\label{ref:ProofProp}

\begin{proof} We assume that the random events  $(\mc{B}_l)_{l\in \{1,\ldots, L\}}$ are independent of the random events $(\mc{A}_l)_{l\in \{1,\ldots, L\}}$. 
\begin{small}
\begin{eqnarray}
 \mc{P}_{\sf{so}} &=&    \mc{P}\bigg(   \bigcup_{i = 1}^L  \mc{B}^c_i \bigg) =  1 -  \mc{P}\bigg(   \bigcap_{i = 1}^L  \mc{B}_i \bigg) \label{eq:SecrecyOutage2} \\
  &=&  1 - \sum_{j=1}^L\mc{P}\bigg(   \bigcap_{i = 1}^j  \mc{B}_i \bigg|\textsf{L} = j \bigg)  \cdot \mc{P}\bigg( \textsf{L} = j \bigg)  \label{eq:SecrecyOutage4} \\ 
  &=&  1 -  \sum_{j=1}^L\mc{P}\bigg(   \bigcap_{i = 1}^j  \mc{B}_i \bigg)  \cdot \mc{P}\bigg( \textsf{L} = j \bigg)  \label{eq:SecrecyOutage5} \\ 
 &=& 1 - \sum_{j=2}^{L-1}    \mc{P}\bigg(   \bigcap_{i = 1}^j  \mc{B}_i  \bigg)  \cdot \Bigg(  \mc{P} \bigg(   \bigcap_{i = 1}^{j-1}  \mc{A}^c_i  \bigg) -  \mc{P}\bigg(  \bigcap_{i = 1}^{j}  \mc{A}^c_i   \bigg) \Bigg)\nonumber \\ 
 &-&\mc{P}\bigg(   \mc{B}_1  \bigg)  \cdot  \mc{P} \bigg(     \mc{A}_1  \bigg) -  \mc{P}\bigg(   \bigcap_{i = 1}^L  \mc{B}_i  \bigg)  \cdot  \mc{P} \bigg(   \bigcap_{i = 1}^{L-1}  \mc{A}^c_i  \bigg) . \label{eq:SecrecyOutage6} 
   \end{eqnarray}
  \end{small}
\eqref{eq:SecrecyOutage2} comes from the properties of the probability $ \mc{P}_{\sf{so}} $.\\ 
\eqref{eq:SecrecyOutage4} comes from the definition of the HARQ-code, if $j$ transmissions occurs, then $ \bigcup_{i = 1}^L  \mc{B}^c_i =  \bigcup_{i = 1}^j  \mc{B}^c_i$.\\
\eqref{eq:SecrecyOutage5} comes from the independence of the events  $(\mc{B}_l)_{l\in \{1,\ldots, L\}}$ with events $(\mc{A}_l)_{l\in \{1,\ldots, L\}}$ hence with transmission number $\textsf{L}$.\\
\eqref{eq:SecrecyOutage6} comes from the probability of having $\textsf{L}$ transmission.
\end{proof}

\section{Proof of Theorem \ref{theo:FeasibleOutOneTrans}}\label{sec:FeasibleOutOneTrans}
\begin{proof}
\textbf{First Point.} Increasing $\textsf{R}$ decreases the connection outage probability and does not affect the secrecy outage probability. Hence we consider the secrecy rate $\textsf{R}=0$. 
\begin{eqnarray}
\mc{P}_{\sf{co}}=  1 - e^{ - \frac{2^{  \textcolor[rgb]{0.00,0.00,0.00}{\textsf{R}_{1}} } - 1}{\gamma_{\textsf{d}}} } \leq  \xi_{\sf{c}} , \quad 
\mc{P}_{\sf{so}} =     e^{ - \frac{2^{\textcolor[rgb]{0.00,0.00,0.00}{\textsf{R}_{1}} } - 1}{\gamma_{\textsf{e}}} } \leq  \xi_{\sf{s}} . \label{eq:1TransSecrecyOutage2}
\end{eqnarray}
$\xi_{\sf{c}}$ and $\xi_{\sf{s}}$ are compatible if there exists $\textsf{R}_{1}$ satisfying \eqref{eq:1TransSecrecyOutage2}, \ie
\begin{eqnarray*}
\log_2\bigg( 1 - \gamma_{\textsf{e}} \cdot  \ln(\xi_{\sf{s}})\bigg) \leq  \textsf{R}_{1} \leq \log_2\bigg( 1- \gamma_{\textsf{d}} \cdot \ln\bigg(1 - \xi_{\sf{c}}\bigg) \bigg).\label{eq:FeasibleRateParameter}
\end{eqnarray*}
The existence of parameter $\textsf{R}_{1} $ is given by the above inequalities and this proves the first point of Theorem \ref{theo:FeasibleOutOneTrans}.\\
 \textbf{Second Point.} The parameter  $\textsf{R}_1$ should satisfy :
\begin{eqnarray*}
\log_2( 1 - \gamma_{\textsf{e}} \cdot \ln(\xi_{\sf{s}}))  \leq \sf{R_1}   \leq\log_2 ( 1 -  \gamma_{\textsf{d}} \cdot \ln(1 - \xi_{\sf{c}}))  - \textsf{R}.
 \end{eqnarray*}
Hence, the parameter $\textsf{R}_1$ exists if and only if:  
\begin{eqnarray*}
 \textcolor[rgb]{0,0,0}{\textsf{R}}  \leq  \log_2 \bigg( \frac{1 -  \gamma_{\textsf{d}} \cdot \ln(1 - \xi_{\sf{c}})}{1 - \gamma_{\textsf{e}} \cdot \ln(\xi_{\sf{s}})}\bigg).
 \end{eqnarray*}
  \end{proof}






\bibliographystyle{IEEEtran}
\bibliography{IEEEabrv,BiblioMael31May17}

\begin{thebibliography}{10}
\providecommand{\url}[1]{#1}
\csname url@samestyle\endcsname
\providecommand{\newblock}{\relax}
\providecommand{\bibinfo}[2]{#2}
\providecommand{\BIBentrySTDinterwordspacing}{\spaceskip=0pt\relax}
\providecommand{\BIBentryALTinterwordstretchfactor}{4}
\providecommand{\BIBentryALTinterwordspacing}{\spaceskip=\fontdimen2\font plus
\BIBentryALTinterwordstretchfactor\fontdimen3\font minus
  \fontdimen4\font\relax}
\providecommand{\BIBforeignlanguage}[2]{{%
\expandafter\ifx\csname l@#1\endcsname\relax
\typeout{** WARNING: IEEEtran.bst: No hyphenation pattern has been}%
\typeout{** loaded for the language `#1'. Using the pattern for}%
\typeout{** the default language instead.}%
\else
\language=\csname l@#1\endcsname
\fi
#2}}
\providecommand{\BIBdecl}{\relax}
\BIBdecl

\bibitem{LeTreustSzczecinskiLabeau13}
M.~Le~Treust, L.~Szczecinski, and F.~Labeau, ``Secrecy \& rate adaptation for
  secure {HARQ} protocols,'' in \emph{Proc. IEEE Information Theory Workshop
  (ITW)}, Sept. 2013, pp. 1--5.

\bibitem{TangLiuSpasojevicPoor09}
X.~Tang, R.~Liu, P.~Spasojevic, and H.~Poor, ``On the throughput of secure
  hybrid-{ARQ} protocols for {G}aussian block-fading channels,'' \emph{{IEEE}
  Trans. Inf. Theory}, vol.~55, no.~4, pp. 1575--1591, Aug. 2009.

\bibitem{shannon-bell-1948}
C.~E. Shannon, ``A mathematical theory of communication,'' \emph{Bell System
  Technical Journal}, vol.~27, pp. 379--423, 1948.

\bibitem{TelatarGallager95}
I.~E. Telatar and R.~G. Gallager, ``Combining queueing theory with information
  theory for multiaccess,'' \emph{{IEEE} J. Sel. Areas Commun.}, vol.~13,
  no.~6, pp. 963--969, Aug. 1995.

\bibitem{CaireTuninettiHARQ2001}
G.~Caire and D.~Tuninetti, ``Throughput of hybrid-{ARQ} protocols for
  {G}aussian collision channel,'' \emph{{IEEE} Trans. Inf. Theory}, vol.~47,
  no.~5, pp. 1971--1988, July 2001.

\bibitem{ZorziRaob96}
M.~Zorzi and R.~R. Rao, ``Throughput performance of {ARQ} selective-repeat with
  time diversity in {M}arkov channels with unreliable feedback,''
  \emph{Wireless Network}, vol.~2, pp. 63--75, 1996.

\bibitem{ZorziRao97}
------, ``Performance of {ARQ} go-back-protocol in {M}arkov channels with
  unreliable feedback,'' \emph{Mobile Networks and Applications}, vol.~2,
  no.~9, pp. 183--193, 1997.

\bibitem{ZorziBorgonovo97}
M.~Zorzi and F.~Borgonovo, ``Performance of capture-division packet access with
  slow shadowing and power control,'' \emph{{IEEE} Trans. Veh. Technol.},
  vol.~46, pp. 687--696, 1997.

\bibitem{Zorzi98}
M.~Zorzi, ``Mobile radio slotted {ALOHA} with capture, diversity and re-
  transmission control in the presence of shadowing,'' \emph{Wireless
  Networks}, vol.~4, pp. 379--388, 1998.

\bibitem{VisotskyYakunTripathiHonigPeterson05}
E.~Visotsky, S.~Yakun, V.~Tripathi, M.~Honig, and R.~Peterson,
  ``Reliability-based incremental redundancy with convolutional codes,''
  \emph{{IEEE} Trans. Commun.}, vol.~53, no.~6, pp. 987--997, June 2005.

\bibitem{PfletschingerNavarro10}
S.~Pfletschinger and M.~Navarro, ``Adaptive {{HARQ}} for imperfect channel
  knowledge,'' in \emph{International ITG Conference on Source and Channel
  Coding (SCC)}, Jan. 2010, pp. 1--6.

\bibitem{UhlemannRasmussenGrantWiberg10}
E.~Uhlemann, L.~Rasmussen, A.~Grant, and P.~Wiberg, ``Optimal
  incremental-redundancy strategy for type-{II} hybrid {ARQ},'' in \emph{{IEEE}
  Inter. Symp. Inf. Theory (ISIT)}, July 2003.

\bibitem{SzczecinskiKhosraviradDuhamelRahman13}
L.~Szczecinski, S.~R. Khosravirad, P.~Duhamel, and M.~Rahman, ``Rate allocation
  and adaptation for incremental redundancy truncated {HARQ},'' \emph{{IEEE}
  Trans. Commun.}, vol.~61, no.~6, pp. 2580--2590, June 2013.

\bibitem{JabiBenjillaliSzczecinskiLabeau16}
M.~Jabi, M.~Benjillali, L.~Szczecinski, and F.~Labeau, ``Energy efficiency of
  adaptive {HARQ},'' \emph{{IEEE} Trans. Commun.}, vol.~64, no.~2, pp.
  818--831, Feb. 2016.

\bibitem{JabiHamssSzczecinskiPiantanida15}
M.~Jabi, A.~E. Hamss, L.~Szczecinski, and P.~Piantanida, ``Multipacket hybrid
  {ARQ}: Closing gap to the ergodic capacity,'' \emph{{IEEE} Trans. Commun.},
  vol.~63, no.~12, pp. 5191--5205, Dec. 2015.

\bibitem{JabiSzczecinskiBenjillaliLabeau15}
M.~Jabi, L.~Szczecinski, M.~Benjillali, and F.~Labeau, ``Outage minimization
  via power adaptation and allocation in truncated hybrid {ARQ},'' \emph{{IEEE}
  Trans. Commun.}, vol.~63, no.~3, pp. 711--723, March 2015.

\bibitem{Larsson16}
P.~Larsson, L.~Rasmussen, and M.~Skoglund, ``Throughput analysis of
  hybrid-{ARQ} -- a matrix exponential distribution approach,'' \emph{{IEEE}
  Trans. Commun.}, vol.~64, no.~1, pp. 416--428, Jan. 2016.

\bibitem{saee}
K.~Nguyen, L.~Rasmussen, A.~Guillen~i Fabregas, and N.~Letzepis, ``{MIMO} {ARQ}
  with multibit feedback: Outage analysis,'' \emph{{IEEE} Trans. Inf. Theory},
  vol.~58, no.~2, pp. 765--779, Feb. 2012.

\bibitem{Lee15}
W.~Lee, O.~Simeone, J.~Kang, S.~Rangan, and P.~Popovski, ``{{HARQ}} buffer
  management: An information-theoretic view,'' \emph{{IEEE} Trans. Commun.},
  vol.~63, no.~11, pp. 4539--4550, Nov. 2015.

\bibitem{Hausl07}
C.~Hausl and A.~Chindapol, ``Hybrid {ARQ} with cross-packet channel coding,''
  \emph{{IEEE} Commun. Lett.}, vol.~11, no.~5, pp. 434--436, May 2007.

\bibitem{Chui07}
J.~Chui and A.~Chindapol, ``Design of cross-packet channel coding with
  low-density parity-check codes,'' in \emph{{IEEE} Information Theory Workshop
  on Information Theory for Wireless Networks}, July 2007, pp. 1--5.

\bibitem{Duyck10}
D.~Duyck, D.~Capirone, C.~Hausl, and M.~Moeneclaey, ``Design of
  diversity-achieving {LDPC} codes for {H-ARQ} with cross-packet channel
  coding,'' in \emph{IEEE 21st Int. Symp. on Personal Indoor and Mobile Radio
  Communications ({PIMRC})}, 2010, pp. 263--268.

\bibitem{Trillingsgaard14}
K.~Trillingsgaard and P.~Popovski, ``Block-fading channels with delayed {CSIT}
  at finite blocklength,'' in \emph{{IEEE} Inter. Symp. Inf. Theory (ISIT)},
  June 2014, pp. 2062--2066.

\bibitem{NguyenTimo15}
K.~D. Nguyen, R.~Timo, and L.~K. Rasmussen, ``Causal-{CSIT} rate adaptation for
  block-fading channels,'' in \emph{{IEEE} Inter. Symp. Inf. Theory (ISIT)},
  June 2015, pp. 351--355.

\bibitem{BenyoussJabiLeTreustSzczecinski16}
A.~Benyouss, M.~Jabi, M.~Le~Treust, and L.~Szczecinski, ``Joint coding/decoding
  for multi-message {HARQ},'' \emph{Proc. of the IEEE Proc. of the Wireless
  Comm. and Networking Conf. (WCNC), Doha, Quatar}, 2016.

\bibitem{JabiBenyoussLeTreustDoraySzczecinski16}
M.~Jabi, A.~Benyouss, M.~Le~Treust, E.~Pierre-Doray, and L.~Szczecinski,
  ``Adaptative cross-packet {HARQ},'' \emph{{IEEE} Trans. Commun.}, vol.~65,
  no.~5, pp. 2022 -- 2035, 2017.

\bibitem{Wyner(Wiretap)1975}
A.~D. Wyner, ``The wire-tap channel,'' \emph{The Bell System Technical
  Journal}, vol.~54, no.~8, pp. 1355--1387, 1975.

\bibitem{Shannon(secrecy)1949}
C.~E. Shannon, ``Communication theory of secrecy systems,'' \emph{Bell System
  Technical Journal}, vol.~28, pp. 656--715, 1949.

\bibitem{CsiszarKorner(BroadcastConf)78}
I.~Csisz\'{a}r and J.~K\"{o}rner, ``Broadcast channels with confidential
  messages,'' \emph{{IEEE} Trans. Inf. Theory}, vol.~24, no.~3, pp. 339--348,
  1978.

\bibitem{LeungHellman(GaussianWiretap)78}
S.~Leung-Yan-Cheong and M.~Hellman, ``The {G}aussian wire-tap channel,''
  \emph{{IEEE} Trans. Inf. Theory}, vol.~24, pp. 451--456, 1978.

\bibitem{BlochBarros11}
M.~Bloch and J.~Barros, \emph{Physical Layer Security-From Information Theory
  to Security Engineering}.\hskip 1em plus 0.5em minus 0.4em\relax Cambridge
  University Press, Oct. 2011.

\bibitem{BlochBarrosRodriguesMcLaughlin08}
M.~Bloch, J.~Barros, M.~R.~D. Rodrigues, and S.~W. McLaughlin, ``Wireless
  information-theoretic security,'' \emph{{IEEE} Trans. Inf. Theory}, vol.~54,
  no.~6, pp. 2515--2534, June 2008.

\bibitem{KhistiTchamkertenWornell08}
A.~Khisti, A.~Tchamkerten, and G.~Wornell, ``Secure broadcasting over fading
  channels,'' \emph{{IEEE} Trans. Inf. Theory}, vol.~54, pp. 2453--2469, 2008.

\bibitem{GopalaLaiElGamal08}
P.~K. Gopala, L.~Lai, and H.~E. Gamal, ``On the secrecy capacity of fading
  channels,'' \emph{{IEEE} Trans. Inf. Theory}, vol.~54, no.~10, pp.
  4687--4698, Oct. 2008.

\bibitem{LiangPoorShamai08}
Y.~Liang, H.~V. Poor, and S.~Shamai, ``Secure communication over fading
  channels,'' \emph{{IEEE} Trans. Inf. Theory}, vol.~54, no.~6, pp. 2470--2492,
  June 2008.

\bibitem{LiYatesTrappeTWC10}
Z.~Li, R.~Yates, and W.~Trappe, ``Achieving secret communication for fast
  {R}ayleigh fading channels,'' \emph{{IEEE} Trans. Wireless Commun.}, vol.~9,
  no.~9, pp. 2792--2799, Sept. 2010.

\bibitem{BlochLaneman13}
M.~R. Bloch and J.~N. Laneman, ``Exploiting partial channel state information
  for secrecy over wireless channels,'' \emph{{IEEE} J. Sel. Areas Commun.},
  vol.~31, no.~9, pp. 1840--1849, Sept. 2013.

\bibitem{RezkiKhistiAlouini14}
Z.~Rezki, A.~Khisti, and M.~S. Alouini, ``On the secrecy capacity of the
  wiretap channel with imperfect main channel estimation,'' \emph{{IEEE} Trans.
  Commun.}, vol.~62, no.~10, pp. 3652--3664, Oct. 2014.

\bibitem{LinJorswieck16}
P.~H. Lin and E.~Jorswieck, ``On the fast fading {G}aussian wiretap channel
  with statistical channel state information at the transmitter,'' \emph{{IEEE}
  Trans. Inf. Forensics Security}, vol.~11, no.~1, pp. 46--58, Jan. 2016.

\bibitem{ZhouMcKayMahamHjorungnes11}
X.~Zhou, M.~R. McKay, B.~Maham, and A.~Hjorungnes, ``Rethinking the secrecy
  outage formulation: A secure transmission design perspective,'' \emph{{IEEE}
  Commun. Lett.}, vol.~15, no.~3, pp. 302--304, March 2011.

\bibitem{MheichLeTreustAlbergeDuhamelSzczecinski14}
Z.~Mheich, M.~Le~Treust, F.~Alberge, P.~Duhamel, and L.~Szczecinski,
  ``Rate-adaptive secure {HARQ} protocol for block-fading channels,'' in
  \emph{22nd European Signal Processing Conference (EUSIPCO)}, Sept. 2014, pp.
  830--834.

\bibitem{MheichLeTreustAlbergeDuhamel16}
Z.~Mheich, M.~Le~Treust, F.~Alberge, and P.~Duhamel, ``Rate adaptation for
  incremental redundancy secure {HARQ},'' \emph{{IEEE} Trans. Commun.},
  vol.~64, no.~2, pp. 765--777, Feb. 2016.

\bibitem{BaldiBianchChiaraluce12}
M.~Baldi, M.~Bianchi, and F.~Chiaraluce, ``Coding with scrambling,
  concatenation, and {HARQ} for the {AWGN} wire-tap channel: A security gap
  analysis,'' \emph{{IEEE} Trans. Inf. Forensics Security}, vol.~7, no.~3, pp.
  883--894, June 2012.

\bibitem{TomasinLaurenti14}
S.~Tomasin and N.~Laurenti, ``Secure {HARQ} with multiple encoding over block
  fading channels: Channel set characterization and outage analysis,''
  \emph{{IEEE} Trans. Inf. Forensics Security}, vol.~9, no.~10, pp. 1708--1719,
  Oct. 2014.

\bibitem{Choi17}
J.~Choi, ``On channel-aware secure {HARQ-IR},'' \emph{{IEEE} Trans. Inf.
  Forensics Security}, vol.~12, no.~2, pp. 351--362, Feb 2017.

\bibitem{HouKramer14}
J.~Hou and G.~Kramer, ``Effective secrecy: Reliability, confusion and
  stealth,'' in \emph{{IEEE} Inter. Symp. Inf. Theory (ISIT)}, June 2014, pp.
  601--605.

\bibitem{GoldfeldCuffPermuter16}
Z.~Goldfeld, P.~Cuff, and H.~H. Permuter, ``Semantic-security capacity for
  wiretap channels of type {II},'' \emph{{IEEE} Trans. Inf. Theory}, vol.~62,
  no.~7, pp. 3863--3879, July 2016.

\bibitem{SenigagliesiBaldiChiaraluce17}
\BIBentryALTinterwordspacing
L.~Senigagliesi, M.~Baldi, and F.~Chiaraluce, ``Semantic security with
  practical transmission schemes over fading wiretap channels,''
  \emph{Entropy}, vol.~19, no.~9, 2017. [Online]. Available:
  \url{http://www.mdpi.com/1099-4300/19/9/491}
\BIBentrySTDinterwordspacing

\bibitem{WangWornellZheng16}
L.~Wang, G.~W. Wornell, and L.~Zheng, ``Fundamental limits of communication
  with low probability of detection,'' \emph{{IEEE} Trans. Inf. Theory},
  vol.~62, no.~6, pp. 3493--3503, June 2016.

\bibitem{Bloch16}
M.~R. Bloch, ``Covert communication over noisy channels: A resolvability
  perspective,'' \emph{{IEEE} Trans. Inf. Theory}, vol.~62, no.~5, pp.
  2334--2354, May 2016.

\bibitem{ZorziRao96}
M.~Zorzi and R.~Rao, ``On the use of renewal theory in the analysis of {ARQ}
  protocols,'' \emph{{IEEE} Trans. Commun.}, vol.~44, no.~9, pp. 1077--1081,
  Sept. 1996.

\bibitem{HarsiniLahoutiLevoratoZorzi11}
J.~S. Harsini, F.~Lahouti, M.~Levorato, and M.~Zorzi, ``Analysis of
  non-cooperative and cooperative type {II} hybrid {ARQ} protocols with {AMC}
  over correlated fading channels,'' \emph{IEEE Transactions on Wireless
  Communications}, vol.~10, no.~3, pp. 877--889, March 2011.

\bibitem{KimChoiBanSung11}
S.~M. Kim, W.~Choi, T.~W. Ban, and D.~K. Sung, ``Optimal rate adaptation for
  hybrid {ARQ} in time-correlated {R}ayleigh fading channels,'' \emph{IEEE
  Transactions on Wireless Communications}, vol.~10, no.~3, pp. 968--979, March
  2011.

\bibitem{ChaitanyaLarsson14}
T.~V.~K. Chaitanya and E.~G. Larsson, ``Adaptive power allocation for {HARQ}
  with {C}hase combining in correlated {R}ayleigh fading channels,'' \emph{IEEE
  Wireless Communications Letters}, vol.~3, no.~2, pp. 169--172, April 2014.

\bibitem{ShiDingMaTam15}
Z.~Shi, H.~Ding, S.~Ma, and K.~W. Tam, ``Analysis of {HARQ-IR} over
  time-correlated {R}ayleigh fading channels,'' \emph{IEEE Transactions on
  Wireless Communications}, vol.~14, no.~12, pp. 7096--7109, Dec 2015.

\bibitem{ElGammalKim(book)11}
A.~E. Gamal and Y.-H. Kim, \emph{Network Information Theory}.\hskip 1em plus
  0.5em minus 0.4em\relax Cambridge University Press, Dec. 2011.

\end{thebibliography}

%

\begin{IEEEbiography}[{\includegraphics[width=1in,height=1.25in,clip,keepaspectratio]{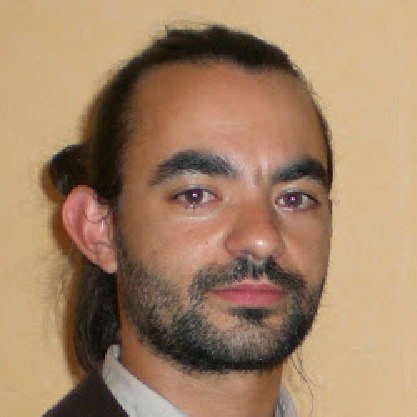}}]{Maël Le Treust} (M'08), earned his Diplôme d'Etude Approfondies (M.Sc.) degree in Optimization, Game Theory \& Economics (OJME) from the Université de Paris VI (UPMC), France in 2008 and his Ph.D. degree from the Université de Paris Sud XI in 2011, at the Laboratoire des signaux et systèmes (joint laboratory of CNRS, Supélec, Université de Paris Sud XI) in Gif-sur-Yvette, France. Since 2013, he is a CNRS researcher at ETIS laboratory UMR 8051, Université Paris Seine, Université Cergy-Pontoise, ENSEA, CNRS, in Cergy, France. In 2012, he was a post-doctoral researcher at the Institut d'électronique et d'informatique Gaspard Monge (Université Paris-Est) in Marne-la-Vallée, France. In 2012-2013, he was a post-doctoral researcher at the Centre Énergie, Matériaux et Télécommunication (Université INRS ) in Montréal, Canada. From 2008 to 2012, he was a Math T.A. at the Université de Paris I (Panthéon-Sorbonne), Université de Paris VI (UPMC) and Université Paris Est Marne-la- Vallée, France. His research interests are strategic coordination, information theory, Shannon theory, game theory, physical layer security and wireless communications.
\end{IEEEbiography}

\begin{IEEEbiography}[{\includegraphics[width=1in,height=1.25in,clip,keepaspectratio]{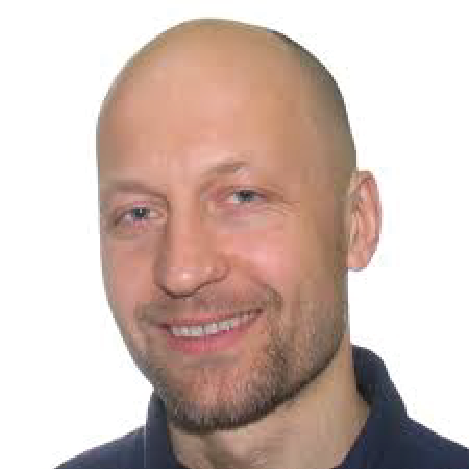}}]{Leszek Szczecinski} (M'98-SM'07), received M.Eng. degree from the Technical University of Warsaw in 1992, and Ph.D. degree from INRS-Telecommunications, Montreal in 1997. 

From 1998 to 2001, he was an Assistant Professor with the Department of Electrical Engineering, University of Chile. He is currently a Professor with INRS, University of Quebec, Canada; 2009-2013 he was an Adjunct Professor with the Electrical and Computer Engineering Department, McGill University. In 2009-2010, he was a Marie Curie Research Fellow with the Laboratory of Signals and Systems, CNRS, Gif-sur-Yvette, France. He co-authored the book ``Bit-Interleaved Coded Modulation: Fundamental, Analysis and Design'' (Wiley, 2015). His research interests include the area of communication theory, modulation and coding, ARQ, wireless communications, and digital signal processing. 
\end{IEEEbiography}
\begin{IEEEbiography}[{\includegraphics[width=1in,height=1.25in,clip,keepaspectratio]{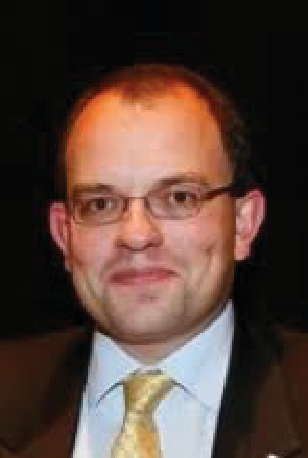}}]{Fabrice Labeau} From January to March 1999, he was a Visiting Scientist with the Department of Signal Processing and Images (TSI), Ecole Na- tionale Supérieure des Télécommunications de Paris, Paris, France. He is currently an Associate Professor with the Department of Electrical and Computer Engineering, McGill University, Montreal, Canada, where he also holds the NSERC/Hydro-Quebec In- dustrial Research Chair in Interactive Information Infrastructure for the power grid. His research in- terests include signal processing and its applications in health, power grids, communications, and compression. He has authored or coauthored over 100 papers in refereed journals and conferences in these areas. Mr. Labeau is or was a Technical Cochair of the 2006 and 2012 Fall IEEE VTC conferences and the 2015 IEEE International Conference on Image Processing. He is also the Executive Vice President and the President-Elect of the IEEE Vehicular Technology Society, a member of the Administrative Committee of the IEEE Sensors Council, and the Chair of the Montreal IEEE
Signal Processing Society Chapter.
\end{IEEEbiography}

\end{document}